\documentclass[11pt]{article}

\usepackage{xspace}
\usepackage{verbatim}
\usepackage{color}
\usepackage{graphicx}
\usepackage{tabularx}

\usepackage{amsmath}
\usepackage[showonlyrefs]{mathtools}
\usepackage{mathstyle}
\usepackage{breqn}
\usepackage{empheq}
\usepackage{amstext,amssymb,amsfonts}
\usepackage{fullpage}
\usepackage{nicefrac}
\usepackage{bm}

\usepackage{mdwlist}

\usepackage{todonotes}

\newcounter{mynotes}
\newcommand{\mnote}[1]{\addtocounter{mynotes}{1}{{}}%
\todo[color=blue!20!white]{[\arabic{mynotes}] \scriptsize  {{\sf {#1}}}}}

\usepackage{nameref}

\usepackage[pagebackref,colorlinks,linkcolor=blue,filecolor = blue,
citecolor = blue, urlcolor = blue,pdfstartview=FitH]{hyperref}
\usepackage[capitalize]{cleveref}

\usepackage{textcomp,setspace}

\usepackage{amsthm}
\usepackage{thmtools}

\declaretheorem[within=section]{theorem}
\declaretheorem[sibling=theorem]{corollary}

\declaretheorem[sibling=theorem]{lemma}
\declaretheorem[sibling=theorem]{claim}

\declaretheorem[sibling=theorem]{definition}
\declaretheorem[sibling=theorem]{Lemma+Definition}
\declaretheorem[sibling=theorem]{remark}
\declaretheorem[sibling=theorem]{observation}
\declaretheorem[sibling=theorem]{conjecture}

\newenvironment{proofof}[1]{\noindent{\bf Proof
of #1:}}{\qed\bigskip}

\crefname{conjecture}{Conjecture}{Conjectures}
\crefname{claim}{Claim}{Claims}
\crefname{remark}{Remark}{Remarks}
\crefname{Lemma+Definition}{Lemma+Definition}{Lemma+Definition}

\newcounter{termcounter}
\renewcommand{\thetermcounter}{\Alph{termcounter}}
\crefname{term}{term}{terms}
\creflabelformat{term}{\textup{(#2#1#3)}}

\makeatletter
\def\term{\@ifnextchar[\term@optarg\term@noarg}
\def\term@optarg[#1]#2{%
  \textup{(#1)}%
  \def\@currentlabel{#1}%
  \def\cref@currentlabel{[][2147483647][]#1}%
  \cref@label[term]{#2}}
\def\term@noarg#1{%
  \refstepcounter{termcounter}%
  \textup{(\thetermcounter)}%
  \cref@label[term]{#1}}
\makeatother

\newcommand{\mrm}[1]{\mathrm {#1}}
\newcommand{\mv}[1]{\mathbf {#1}}
\newcommand{\msf}[1]{\mathsf {#1}}
\newcommand{\mcl}[1]{\mathcal {#1}}

\newcommand{\nfrac}{\nicefrac}

\newcommand{\ignore}[1]{}

\newcommand{\bits}{\{0,1\}}

\newcommand{\Brac}[1]{\left[#1 \right]}
\newcommand{\set}[1]{\left\{#1\right\}}

\definecolor{DSred}{rgb}{1,0,0}

\renewcommand{\leq}{\leqslant}
\renewcommand{\geq}{\geqslant}
\renewcommand{\ge}{\geqslant}
\renewcommand{\le}{\leqslant}
\renewcommand{\epsilon}{\varepsilon}
\newcommand{\eps}{\epsilon}

\newcommand{\R}{\mathbb{R}}
\newcommand{\C}{\mathbb{C}}
\newcommand{\Z}{\mathbb{Z}}
\newcommand{\N}{\mathbb{N}}
\newcommand{\F}{\mathbb{F}}
\newcommand{\T}{\mathbb{T}}
\newcommand{\U}{\mathbb{U}}

\newcommand{\cA}{\mathcal A}
\newcommand{\cB}{\mathcal B}
\newcommand{\cC}{\mathcal C}

\newcommand{\cG}{\mathcal G}

\newcommand{\cL}{\mathcal L}

\newcommand{\cP}{\mathcal P}

\newcommand{\Esymb}{{\bf E}}

\newcommand{\Psymb}{{\bf Pr}}

\DeclareMathOperator*{\E}{\Esymb}

\DeclareMathOperator*{\ProbOp}{\Psymb}

\renewcommand{\Pr}{\ProbOp}

\newcommand{\Ex}[1]{\E\Brac{#1}}

\newcommand{\expo}[1]{{\mathsf{e}\left(#1\right)}}

\title{Every locally characterized affine-invariant property is testable}
\author{Arnab Bhattacharyya\thanks{DIMACS \& Rutgers
    University. Email: \texttt{arnabb@dimacs.rutgers.edu}.}
 \and Eldar Fischer\thanks{Technion. Email: \texttt{eldar@cs.technion.ac.il}.}
 \and Hamed Hatami\thanks{McGill University. Email: \texttt{hatami@cs.mcgill.ca}.}
 \and
  Pooya Hatami\thanks{University of Chicago. Email: \texttt{pooya@cs.uchicago.edu}.}
 \and Shachar Lovett\thanks{Institute for Advanced Study and
   University of California, San Diego. Email: \texttt{slovett@math.ias.edu}.}
}

\begin{document}
\maketitle

\begin{abstract}
Let $\F = \F_p$ for any fixed prime $p\geq 2$.
An affine-invariant property is a property of functions on $\F^n$ that
is closed under taking affine transformations of the domain. We prove that all affine-invariant properties that have local
characterizations are testable.  In fact, we give a proximity-oblivious
test for any such property $\cP$, meaning that given an input function $f$, we make a constant number of queries to
$f$, always accept if $f$ satisfies $\cP$, and otherwise reject with probability larger than a positive number
that depends only on the distance between $f$ and $\cP$. More generally, we show
that any affine-invariant property that is closed under taking
restrictions to subspaces and has bounded complexity is testable.

We also prove that any property that can be described as the property
of being decomposable into a known structure of low-degree polynomials is
locally characterized and is, hence, testable. For example, whether a
function is a product of two degree-$d$ polynomials, whether a
function splits into a product of $d$ linear polynomials, and whether
a function has low rank are all examples of degree-structural
properties and are therefore locally characterized.

Our results use a new Gowers inverse theorem by Tao and Ziegler
for low characteristic fields that decomposes any polynomial with
large Gowers norm into a function of a small number of low-degree {\em non-classical
  polynomials}. We establish a new equidistribution result for high rank
non-classical polynomials that drives the proofs of both the testability results and the local characterization of degree-structural properties.
\end{abstract}

\section{Introduction}\label{sec:intro}
The field of property testing, as initiated by \cite{BLR,BFL} and
defined formally by \cite{RS,GGR}, is the study of algorithms that
query their input a very small number of times and with high
probability decide correctly whether their input satisfies a given property or
is ``far'' from satisfying that property.  A property is called {\em
testable}, or sometimes {\em strongly testable} or {\em locally
testable}, if the number of queries can be made independent of the
size of the object without affecting the correctness probability.
Perhaps surprisingly, it has been found that a large number of
natural properties satisfy this strong requirement; see e.g.\
the surveys \cite{FischerSurvey, RubinfeldICM,
RonSurvey09, SudanSurvey} for a general overview.

The focus of our work is on testing properties of multivariate
functions over finite fields. Fix a prime $p \geq 2$ and an integer $R
\geq 2$ throughout. Let $\F = \F_p$. We consider properties of
functions $f : \F^n \to \set{1, \dots, R}$. Our main result shows that any such
property that is invariant with respect to affine transformations on
$\F^n$ and
that is locally characterized is testable.  Furthermore, we show that
a large class of natural algebraic properties whose query complexity
had not been previously studied are locally characterized
affine-invariant properties and are, hence, testable. Our results
constitute an exact characterization of proximity-obliviously testable
properties, the most common notion of testability considered for
algebraic properties. In the rest of this section, we motivate and
describe our results in more detail.

\subsection{Testability and  Invariances}
Let $[R]$ denote the set $\set{1, \dots, R}$.
Given a property $\cP$ of functions in $\{\F^n \to [R] \ | \ n \in \Z_{\ge 0}\}$,
we say that $f : \F^n \to [R]$ is {\em $\eps$-far} from $\cP$ if
$$\min_{g \in \cP} \Pr_{x \in \F^n}[f(x) \neq g(x)] > \eps,$$
and we say that it is {\em $\eps$-close} otherwise.

\begin{definition}[Testability]\label{testable}
A property $\cP$ is said to be {\em testable} (with one-sided error)
if there are functions $q: (0,1) \to \Z_{> 0}$, $\delta: (0,1) \to (0,1)$,
and an algorithm $T$ that, given as input a parameter $\eps > 0$ and oracle
access to a function $f: \F^n \to [R]$, makes at most $q(\eps)$
queries to the oracle for $f$, always accepts if $f \in \cP$ and
rejects with probability at least $\delta(\eps)$ if $f$ is $\eps$-far
from $\cP$. If, furthermore, $q$ is a constant function, then $\cP$ is
said to be {\em proximity-obliviously testable (PO testable)}.
\end{definition}

The term proximity-oblivious testing is coined by Goldreich and Ron in~\cite{MR2792371}.
As an example of a testable (in fact, PO testable) property, let us recall the famous result by Blum, Luby and Rubinfeld
\cite{BLR} which initiated this line of research. They showed
that linearity of a function $f: \F^n \to \F$ is testable by a test
which makes $3$ queries. This test accepts if $f$ is linear and rejects with
probability $\Omega(\eps)$ if $f$ is $\eps$-far from linear.

Linearity, in addition to being testable, is also an example of a
{\em linear-invariant} property. We say that a property $\cP
\subseteq \{\F^n \to [R]\}$ is linear-invariant if it is the case
that for any $f \in \cP$ and for any linear transformation $L:
\F^n \to \F^n$, it holds that $f\circ L \in \cP$.  Similarly, an
{\em affine-invariant} property is closed under composition with
affine transformations $A: \F^n \to \F^n$ (an affine transformation
$A$ is of the form $L+c$ where $L$ is linear and $c$ is a constant).
The property of a function $f: \F^n \to \F$ being affine is testable
by a simple reduction to \cite{BLR}, and is itself affine-invariant.
 Other well-studied
examples of affine-invariant (and hence, linear-invariant) properties
include Reed-Muller codes (in other words, bounded degree
polynomials)
\cite{BFL, BFLS,FGLSS,RS,AKKLR}
and Fourier sparsity
\cite{GOSSW}.
In fact, affine invariance seems to be a common feature of most interesting
properties that one would classify as ``algebraic''.  Kaufman and
Sudan in
\cite{KS08}
made
explicit note of this phenomenon and initiated a general study of the testability of
affine-invariant properties (see also~\cite{MR2863292}). In particular, they asked for necessary
and sufficient conditions for the testability of affine-invariant
properties.
{
\subsection{Locally Characterized Properties}\label{sec:degstructintro}

The result summarized in the title of this paper gives a necessary and
sufficient condition for affine-invariant
properties to be PO testable. Let us first see why {``local characterization''} is a
necessary condition for PO testability.

For a PO testable property $\cP$, if a function
$f$ does not satisfy $\cP$, then by \cref{testable}, the tester rejects $f$ with positive
probability. Since the test always accepts functions with the property, there must be $q$ points $x_1, \dots, x_q \in
\F^n$ that form a witness for non-membership in $\cP$. These
are the queries that cause the tester to reject. Thus, denoting
$\sigma = (f(x_1), \dots, f(x_q)) \in [R]^q$, we say that $\cC = (x_1,
x_2, \dots, x_q; \sigma)$ forms a {\em $q$-local constraint} for
$\cP$. This means that whenever the constraint is violated by a
function $g$, i.e., $(g(x_1), \dots, g(x_q)) = \sigma$, we know that $g$ is not in $\cP$.   A
property $\cP$ is {\em $q$-locally characterized} if there
exists a collection of $q$-local constraints $\cC_1, \dots, \cC_m$ such that
$g \in \cP$ if and only if none of the constraints  $\cC_1, \dots,
\cC_m$ are violated.  It follows from the above discussion that if $\cP$ is
PO testable with $q$ queries, then $\cP$ is
$q$-locally characterized. We say $\cP$ is {\em locally characterized}
if it is $q$-locally characterized for some constant $q$.

We now give some examples of locally characterized
affine-invariant properties. Consider the property
of being affine. It is $4$-locally characterized because a
function $f$ is affine if and only if $f(x)-f(x+y)-f(x+z)+f(x+y+z) = 0$
for every $x,y,z\in \F^n$. Note that this characterization
automatically suggests a $4$-query test: pick random $x, y, z \in
\F^n$ and check whether the identity holds or not for that choice of
$x, y, z$.
More generally, consider the property of being a
polynomial of degree at most $d$, for some fixed integer $d > 0$.
The property is known to be PO testable due to
independent work
of \cite{KR06,  JPRZ}, and their test is based upon a $p^{\lceil
\frac{d+1}{p-1}\rceil}$-local characterization. Again, the test is
simply to pick a random constraint and check if it is violated.

Indeed, for any $q$-locally characterized property $\cP$ defined by
constraints $\cC_1, \dots, \cC_m$, one can design the following
$q$-query test: choose a constraint $\cC_i$ uniformly at random and
reject only if the input function violates $\cC_i$. Clearly, if the
input function $f$ is in $\cP$, the test always accepts. The
question is the probability with which a function $\eps$-far from
$\cP$ is rejected. We show that for affine-invariant properties, this test always rejects with
probability bounded away from zero for every constant $\eps > 0$.
\begin{theorem}\label{thm:main}
Every $q$-locally characterized affine-invariant property is
proximity-obliviously testable   with $q$ queries.
\end{theorem}

\subsection{Subspace Hereditary Properties}
Just as a necessary condition for PO testability is local
characterization, one can formulate a natural condition that is
(almost) necessary for testability in general. In the context of
affine-invariant properties, the condition can be succinctly stated as
follows:
\begin{definition}[Subspace hereditary properties]
An affine-invariant property $\cP$ is said to be {\em (affine) subspace
  hereditary} if for any  $f: \F^n \to [R]$ satisfying $
\cP$, the restriction of $f$ to any affine subspace of $\F^n$ also
satisfies $\cP$.
\end{definition}
In \cite{BGS10}, it is shown that {every} affine-invariant property
testable by a ``natural'' tester is very ``close'' to a
subspace hereditary property\footnote{We omit the technical definitions
of ``natural'' and ``close'', since they are unimportant
here. Informally, the behavior of a ``natural'' tester is independent of the
size of the domain and ``close'' means that the property deviates from
an actual affine subspace hereditary property on functions over a
finite domain. See \cite{BGS10} for details, or \cite{AS08} for
the analogous definitions in a graph-theoretic
context.}. Thus, if we gloss over some technicalities,  subspace
hereditariness is a necessary condition for testability. In the
opposite direction, \cite{BGS10} conjectures the following:
\begin{conjecture}[\cite{BGS10}]\label{conj:main}
Every subspace hereditary property is testable.
\end{conjecture}
Resolving Conjecture \ref{conj:main} would yield a
combinatorial {\em characterization} of the (natural) one-sided testable
affine-invariant properties, similar to the characterization for
testable dense graph properties \cite{AS08}.
Although we are not yet able to confirm or refute the full
\cref{conj:main}, we can show testability if we make an additional
assumption of ``bounded complexity'', defined formally in Section~\ref{sec:complexity}.
\begin{theorem}[Informal]\label{thm:main2}
Every subspace hereditary property of ``bounded complexity'' is testable.
\end{theorem}
We will formally define the notion of complexity later on in Section~\ref{sec:complexity}, but for now, it suffices
to know that it is an integer that we will associate with each
property (independent of $n$). Also, $q$-locally characterized
properties are of complexity at most $q$.
All natural affine-invariant properties
that we know of have bounded complexity (in fact, most are
locally characterized). So, the subspace
hereditary properties not covered by \cref{thm:main2} seem to be mainly
of theoretical interest.

\subsection{Degree-structural Properties}\label{sec:introdegstruct}

The conditions required in \cref{thm:main} and \cref{thm:main2} are very general, and so, we expect that they are
satisfied by many interesting algebraic properties. This, in fact, turns out to be the case. We show that
a class of properties that we call {\em degree-structural} are all locally characterized and are, hence, testable by \cref{thm:main}.
We give the definition below in~\cref{def:weakdegstruct}. First let us list some examples of degree-structural properties. Let
$d$ be a fixed positive integer. Each of the following definitions defines a degree-structural property.
\begin{itemize}
\item
\textbf{Degree $\leq d$:} The degree of the function $F: \F^n \to \F$ as a polynomial is at most $d$;
\item
\textbf{Splitting:} A function $F: \F^n \to \F$ {\em splits} if it
can be written as a product of at most $d$ linear functions;
\item
\textbf{Factorization:} A function $F: \F^n \to \F$ {\em factors} if
$F = GH$ for polynomials $G, H: \F^n \to \F$ such that
$\deg(G) \leq d-1$ and $\deg(H) \leq d-1$;
\item
\textbf{Sum of two products:} A function $F: \F^n \to \F$ is a
{\em sum of two products} if there are polynomials $G_1, G_2, G_3, G_4$ such
that $F=G_1G_2+G_3G_4$ and $\deg(G_i) \leq d-1$ for $i \in \set{1,2,3,4}$;
\item
\textbf{Having square root:} A function $F: \F^n \to \F$ {\em has a square
root} if $F = G^2$ for a polynomial $G$ with $\deg(G) \leq d/2$;

\item
\textbf{Low $d$-rank:} for a fixed integer $r > 0$, a function $F:
\F^n \to \F$ {\em has $d$-rank at most $r$} if there exist polynomials $G_1,
\dots, G_r: \F^n \to \F$ of degree $\leq d-1$ and a function $\Gamma:
\F^r \to \F$ such that $F = \Gamma(G_1, \dots, G_r)$.
\end{itemize}

In fact, roughly speaking, any property that can be described as the
property of decomposing into a known structure of low-degree
polynomials is degree-structural.

\begin{definition}[Degree-structural property] \label{def:weakdegstruct}
Given an integer $c > 0$, a vector of non-negative integers
$\mathbf{d} = (d_1,\dots,d_c) \in \Z_{\geq 0}^c$,  and a function
$\Gamma : \F^c \to \F$,  define the {\em $(c,\mathbf{d},\Gamma)$-structured
  property} to be the collection of functions $F: \F^n \to \F$ for which
there exist polynomials  $P_1,\dots,P_c: \F^n \to \F$ satisfying
$F(x) = \Gamma(P_1(x), \dots, P_c(x))$ for all $x \in \F^n$ and $\deg(P_i)
\leq d_i$ for all $i \in [c]$.

We say a property $\mcl{P}$ is {\em degree-structural} if there exist
integers $\sigma, \Delta >0$ and a set of tuples $S \subset \set{(c,\mv{d},\Gamma)
  \mid c \in [\sigma], \mv{d}\in [0,\Delta]^c, \Gamma: \F^c \to \F}$, such
that a function $F: \F^n \to \F$ satisfies $\mcl{P}$ if
and only if $F$ is $(c,\mv{d},\Gamma)$-structured for some
$(c,\mv{d},\Gamma) \in S$. We call $\sigma$ the {\em scope} and
$\Delta$ the {\em max-degree} of the degree-structural property $\mcl{P}$.
\end{definition}

It is straightforward to see that the examples above satisfy this definition. Our main
result for degree-structural properties is the following:

\begin{theorem}\label{thm:degstruct}
Every degree-structural property with bounded scope and max-degree is
a locally characterized affine-invariant property.
\end{theorem}

Combining \cref{thm:degstruct} with \cref{thm:main} implies PO testability for all degree-structural properties.

\subsection{Formal version of the Main Result}
In this section, we describe our main result, \cref{thm:main2},
rigorously. \cref{thm:main} follows as a corollary.
We first need to set up some notions. Just as a locally characterized
property can be described by a list of constraints, subspace
hereditary properties can also be described similarly, but here, the size of the list can
be infinite. For affine-invariant properties, we can represent the
constraints in a very special form, as ``induced affine constraints''. We
first describe these, then define the notion of complexity, and
finally state the theorem.

\subsubsection{Affine constraints}

A {\em linear form on $k$ variables} is a vector $L = (w_1, w_2,
\dots, w_k) \in \F^k$ that is interpreted as a function from
$(\F^n)^k$ to $\F^n$ via the map $(x_1, \dots, x_k) \mapsto w_1 x_1 +
w_2 x_2 + \cdots + w_k x_k$. A linear form $L = (w_1, w_2, \dots,
w_k)$ is said to be {\em affine} if $w_1 =1$. From now, linear forms
will always be assumed to be affine.

We specify a partial order $\preceq$ among affine forms. We say $(w_1,
\dots, w_k) \preceq (w_1',\dots, w_k')$ if $|w_i| \leq |w_i'|$ for all
$i \in [k]$, where $|\cdot|$ is the obvious map from $\F$ to $\set{0,
  1, \dots, p-1}$. An affine constraint is a collection of affine forms, with the added
technical restriction of being downward-closed with respect to $\preceq$. For future references
we state this as the following definition.

\begin{definition}[Affine constraints]\label{defaffine}
An {\em affine constraint of size $m$ on $k$ variables} is a tuple
$A = (L_1, \dots, L_m)$ of $m$ affine forms $L_1, \dots, L_m$ over
$\F$ on $k$ variables, where:
\begin{enumerate*}
\item[(i)]
$L_1(x_1, \dots, x_k) = x_1$;

\item[(ii)]
If $L$ belongs to $A$ and $L' \preceq L$, then $L'$ also belongs to $A$.
\end{enumerate*}
\end{definition}

Any subspace hereditary property can be described using affine
constraints and forbidden patterns, in the following way.
\begin{definition}[Properties defined by induced affine constraints]
\
\begin{itemize}
\item
An {\em induced affine constraint of size $m$ on $\ell$ variables} is
a pair $(A,\sigma)$ where $A$ is an affine constraint of size $m$ on
$\ell$ variables and $\sigma \in [R]^m$.
\item
Given such an induced affine constraint $(A,\sigma)$, a function $f:
\F^n \to [R]$ is said to be {\em $(A,\sigma)$-free} if there exist no
$x_1, \dots, x_\ell \in \F^n$ such that $(f(L_1(x_1, \dots, x_\ell)),
\dots, f(L_m(x_1,\dots,x_\ell))) = \sigma$. On the other hand, if such
$x_1, \dots, x_\ell$ exist, we say that {\em $f$ induces $(A,\sigma)$ at
  $x_1, \dots, x_\ell$}.
\item
Given a (possibly infinite) collection $\cA = \{(A^1,\sigma^1),
(A^2, \sigma^2), \dots, (A^i,\sigma^i),\dots\}$ of induced affine constraints, a function $f: \F^n
\to [R]$ is said to be {\em $\cA$-free} if it is
$(A^i,\sigma^i)$-free for every $i \geq 1$.
\end{itemize}
\end{definition}

As an example consider the property of having degree at most $1$ as a polynomial, for function $F: \F^n \to \F$. It is easy to see that $F$ satisfies this property if and only if $F(x_1)-F(x_1+x_2)-F(x_1+x_3)+F(x_1+x_2+x_3)=0$ for all $x_1,x_2,x_3 \in \F$. Consequently the property can be defined by the set of induced affine constraints that forbid any values for $F(x_1),F(x_1+x_2),F(x_1+x_3),F(x_1+x_2+x_3)$ that do not satisfy the identity $F(x_1)-F(x_1+x_2)-F(x_1+x_3)+F(x_1+x_2+x_3)=0$.

The connection between affine subspace hereditariness and affine
constraints is given by the following simple observation.
\begin{observation}
An affine-invariant property $\cP$ is subspace hereditary if
and only if it is equivalent to the property of $\cA$-freeness for
some fixed collection $\cA$ of induced affine constraints.
\end{observation}
\begin{proof}
Given an affine invariant property $\cP$, a simple (though
inefficient) way to obtain the set $\cA$ is to let it be the
following: For every $n$ and a function $f:\F^n \to [R]$ that is not in $\cP$,
we include in $\cA$ the constraint $(A_f,\sigma_f)$, where $A_f$ is
indexed by members of $\F^n$ and contains
$\{L_z(X_1,\ldots,X_{n+1})=X_1+\sum_{i=1}^nz_iX_{i+1}:z=(z_1,\ldots,z_n)\in\F^n\}$,
and $\sigma_f$ is just set to $f$.

Setting $X_1=0$ and $X_{i+1}$ to  the $i$th standard vector $e_i$ for every $i \in [n]$ shows that $f$ is not $(A_f,\sigma_f)$-free.
Hence the property defined by $\cA$ is contained in $\cP$. The containment in the other direction follows from $\cP$ being affine-invariant and hereditary.

The other direction of the observation is trivial.
\end{proof}

\subsubsection{Complexity of linear forms \label{sec:complexity}}
Green and Tao, in their seminal work on arithmetic progressions in
primes, introduced the following notion of complexity of linear forms.

\begin{definition}[Cauchy-Schwarz complexity, \cite{GT06}]\label{def:cplx}
Let $\mathcal{L} = \{L_1,\dots,L_m\}$ be a set of linear forms. The
{\em (Cauchy-Schwarz) complexity of $\mathcal{L}$} is the minimal $d$
such that the following holds. For every $i \in [m]$, we can partition
$\{L_j\}_{j \in [m]\setminus \{i\}}$ into $d+1$ subsets such that
$L_i$ does not belong to the linear span of any subset.
\end{definition}

If  $\cL = \{L_1,\dots,L_m\}$ contains two linear forms that are multiples of each other (that is $L_i = \lambda L_j$ for $i\neq j$ and $\lambda \in \F$), then the complexity of $\cL$ is infinity. Otherwise its complexity is at most $|\cL|-2$. Note that sets of affine linear forms are always of finite complexity. The following lemma can be proved using iterated applications of the classical Cauchy-Schwarz inequality. It explains the term ``Cauchy-Schwarz complexity'', and illustrates its importance. 

\begin{lemma}[Counting Lemma, \cite{GT06}]\label{gowerscount}
Let $f_1, \dots, f_m : \F^n \to [-1,1]$. Let $\mathcal{L} =
\{L_1,\dots,L_m\}$ be a system of $m$ linear forms in $\ell$ variables
of complexity $d$. Then:
$$\left|\E_{x_1,\dots,x_\ell \in \F^n} \left[ \prod_{i=1}^m
    f_i(L_i(x_1,\dots,x_\ell))\right] \right| \leq \min_{i \in [m]}\|f_i\|_{U^{d+1}}.$$
\end{lemma}

Finally, given a collection $\cA = \set{(A^1,\sigma^1), (A^2, \sigma^2),
  \dots, (A^i,\sigma^i)}$ of induced affine constraints, we say that $\cA$
is of {\em complexity $\leq d$} if for each $i$, the collection of
affine forms $A^i$ is of complexity $\leq d$ according to
\cref{def:cplx}.

\subsubsection{Statement of the main result}

\begin{theorem}[Main theorem]\label{thm:main3}
For any integer $d>0$ and (possibly infinite) fixed
collection $\cA = \{(A^1,\sigma^1),$ $(A^2, \sigma^2),$
$\dots,$ $(A^i,\sigma^i), \dots \}$ of induced affine constraints, each of
complexity $\leq d$, there are functions
$q_\cA:(0,1)\to \Z^+$, $\delta_\cA: (0,1) \to (0,1)$ and  a tester $T$
which, for every $\eps > 0$, makes $q_\cA(\eps)$ queries, accepts
$\cA$-free functions and rejects functions $\eps$-far from $\cA$-free
with probability at least $\delta_\cA(\eps)$. Moreover, $q_\cA$ is a
constant if $\cA$ is of finite size.
\end{theorem}

We do not have any explicit bounds on $\delta_\cA$ because the analysis depends on previous work based on
ergodic theory. It would of course be interesting to have explicit
bounds for some of the properties described in \ref{sec:degstructintro}.

Let us finally note that Theorem \ref{thm:main} is quite nontrivial
even  if $\cA$ consists only of a single induced affine constraint of complexity
greater than $1$. Indeed, previously it was not known how to show
testability in this case. A more detailed account of previous work is given in \cref{sec:past}.

\subsection{Overview of Proofs}

\subsubsection{Testability}
Let us now give an overview of our proof of Theorem
\ref{thm:main3}. For simplicity of exposition, assume for now that
$\cA$ consists only of a single induced affine constraint $(A,\sigma)$
where $A$ is the tuple of affine linear forms $(L_1,\dots,L_m)$, each over
$\ell$ variables, and $\sigma \in [R]^m$. Let $d$ be the complexity of
the constraint. For $i \in [R]$, let
$f^{(i)}:\F^n \to \bits$ be the indicator function for the set
$f^{-1}(\{i\})$. Our goal will  be
to show that, when $f$ is $\eps$-far from $(A,\sigma)$-free, then:
\begin{equation}\label{eqn:avg}
\E_{x_1,\dots,x_\ell}\left[f^{(\sigma_1)}(L_1(x_1,\dots,x_\ell)) \cdot
f^{(\sigma_2)}(L_2(x_1,\dots,x_\ell)) \cdot \ldots \cdot
f^{(\sigma_m)}(L_m(x_1,\dots,x_\ell))  \right] \geq \delta(\epsilon),
\end{equation}
for some function $\delta: (0,1)\to (0,1)$. If \cref{eqn:avg} is true, then a valid test would be to
simply pick $\ell$ points uniformly at random and reject only if
$f(L_1(x_1, \dots, x_\ell)) = \sigma_1, \dots, f(L_m(x_1, \dots,x_\ell)) = \sigma_m$.

Studying averages of products, as in (\ref{eqn:avg}), has been crucial
to a wide range of problems in additive combinatorics and analytic
number theory. Szemer\'edi's theorem about the density of arithmetic
progressions in subsets of the integers is a classic
example. Szemer\'edi's work \cite{Szem75} arguably initiated such
questions in additive combinatorics, but the major development which
led to a more systematic understanding of these averages was Gowers'
definition of a new notion of uniformity in a Fourier-analytic proof
for Szemer\'edi's theorem \cite{Gow01}. In particular, Gowers introduced
the {\em Gowers norm} $\|\cdot\|_{U^{d+1}}$,
which allows us to say the following about (\ref{eqn:avg}): If
$\|f_1\|_{U^{d+1}} < \eps$, $f_2,\dots,f_m$ are arbitrary functions
that are bounded inside $[-1,1]$, and $L_1,\dots,L_m$ are
linear forms of complexity at most $d$, then
$$\left|\E_{x_1,\dots,x_\ell \in \F^n} \left[ \prod_{i=1}^m f_i(L_i(x_1,\dots,x_\ell))\right]\right| \le \eps.$$

This observation leads to the study of {\em decomposition theorems},
that express an arbitrary function $f$ as a sum of two functions $g$
and $h$, where $g$ is ``structured'' in a sense we describe soon and $h$ has low $(d+1)$-th order Gowers norm.  Decomposing each $f^{(\sigma_i)}$ in
this way into $g^{(\sigma_i)}$ and $h^{(\sigma_i)}$, substituting into
\cref{eqn:avg} and expanding, we get inside the expectation a sum of
$2^m$ terms. All these terms, except one, contain some $h^{(\sigma_i)}$ in the
product and  can be bounded by the above mentioned property of the
Gowers norm. In fact, we can make the Gowers norm small enough that we can
effectively discard all these terms inside the expectation. The term
remaining is the product of the ``structured'' functions,
\begin{equation}\label{eqn:avg2}
\E_{x_1,\dots,x_\ell}\left[g^{(\sigma_1)}(L_1(x_1,\dots,x_\ell))
g^{(\sigma_2)}(L_2(x_1,\dots,x_\ell)) \ldots
g^{(\sigma_m)}(L_m(x_1,\dots,x_\ell))  \right],
\end{equation}
 and the goal is to lower-bound  this expectation.

To describe the structure of $g$, let us go over how the
decomposition into $g$ and $h$ is obtained. Given an arbitrary function $f$,
if $\|f\|_{U^{d+1}}$ is small, then we are already done. Otherwise, we
repeatedly apply the {\em Gowers inverse theorem} to find a finite
collection of polynomials $P_1, \dots, P_C$ of degree $\leq d$ such
that $f = \Gamma(P_1, \dots, P_C) + h$, where $\|h\|_{U^{d+1}}$ is
small and $\Gamma$ is some function. But there is a catch in this
nice-looking structural theorem! If $p>d$, $P_1,\dots, P_C$ are indeed
``classical'' degree-$d$ $\F$-valued polynomials over $\F^n$. However,
in our setting, where $p$ is a fixed small constant, such a
decomposition may no longer hold. Indeed,  \cite{GT07, LMS} proved
that if $f$ equals the symmetric degree-4 polynomial and $d=3$, we
have an explicit counterexample to such a claim. Fortunately,
Bergelson, Tao and
Ziegler \cite{BTZ10, TZ, TZ11} showed that it is possible to salvage the
decomposition theorem by replacing ``classical''
$\F$-valued polynomials by ``non-classical'' polynomials. These polynomials  may take values
in $\Z_{p^k}$ for some integer $k$.
More precisely, a non-classical polynomial of degree $d$ is a function
$P$ from $\F^n$ to $\Z_{p^k}$ such that the $(d+1)$-th
order derivative of $P$ is zero. The integer $k-1$ is called the
``depth'' of $P$. Classical polynomials have depth $0$.

We use the result of \cite{TZ11} to obtain non-classical polynomials
$P_1, \dots, P_C$ of degree $\leq d$ such that each $g^{(\sigma_i)} = \Gamma_i(P_1,
\dots, P_C)$ for some function $\Gamma_i$. We return now to the goal
of lower-bounding \cref{eqn:avg2}. By a sequence of steps already
introduced in \cite{BGS10} and \cite{BFL12} (inspired by similar
techniques on graph property testing in \cite{AFKS,AS08a, AS08}), we reduce
to the problem of lower-bounding the probability
$$\Pr_{x_1, \dots, x_\ell}\left[\bigwedge_{i \in [C], j \in [m]}
  P_i(L_j(x_1, \dots, x_\ell)) = b_{i,j}\right]
$$
where each $b_{i,j}$ is an arbitrary fixed element in the range of
$P_i$. That is, we want to show that the polynomials $\{P_i \circ L_j
\mid i \in [C], j\in [m]\}$ behave like independent random variables
distributed nearly uniformly on their range. Of course, this cannot be
completely true. For example, if $P_i$ is linear, $P_i(x_1 + x_2 + x_3) -
P_i(x_1 + x_2) - P_i(x_1 + x_3) + P_i(x_1)$ is identically zero and
so, $\{P_i(x_1 + x_2 + x_3), P_i(x_1 + x_2), P_i(x_1 + x_3),
P_i(x_1)\}$  are correlated. Moreover, because the polynomials are
non-classical, $pP$ is a non-constant polynomial of lower degree than
$P$ and satisfies other identities not satisfied by $P$ itself.
What we show is that if the collection of polynomials $P_1, \dots,
P_C$ is of high {\em rank}, then besides
correlations which are forced by the degree and depth of the
polynomials, there are no other dependencies. This equidistribution
result for high rank non-classical polynomials is the technical crux
of our work. Our proof technique is very different from the similar
equidistribution claim in \cite{HL11,HL11b} for classical polynomials, since
that proof uses the monomial structure of classical polynomials.

Let us briefly describe what we mean by a high rank collection of non-classical polynomials
$P_1, \dots, P_C$. We say that the rank of the
collection is $r$ if there exist integers $\lambda_1, \dots,
\lambda_C$ such that
$\lambda_1 P_1, \dots, \lambda_CP_C$ are not all identically zero but
$\sum_{i=1}^C\lambda_i P_i = \Gamma(Q_1, \dots,
Q_r)$ for some $r$ polynomials $Q_1, \dots, Q_r$ each of degree $< \max_i \deg(\lambda_i
P_i)$ and some function $\Gamma$. So, if the rank of a collection of
polynomials is high, that means that no linear combination of the
polynomials, unless it is trivially zero, has an explanation
in terms of a small number of lower degree polynomials. Intuitively, a high rank
collection of degree $d$ polynomials is like a random or generic
collection of degree $d$ polynomials. It does not have unexpected low-degree
correlations, and it is robust to common operations such as taking
projections or multiplying by constants or taking derivatives.

This finishes the high-level overview of the proof, although
there are some additional issues that we have swept under the rug.
One problem is that the decomposition theorem actually decomposes a
given function $f$ to a sum of  three functions $f_1, f_2, f_3$, not
into two functions $g$ and $h$ as in the description above. The functions $f_1$
and $f_2$ correspond to $g$ and $h$, respectively, and $f_3$ is an
additional function that has low $L^2$-norm. Now, the closeness to
equidistribution of the non-classical polynomials $P_1,\dots, P_C$
describing $f_1$ and the smallness of the Gowers norm for $f_2$
can be made arbitrarily small as a function of  $C$
and are thus, essentially negligible for the purposes of the proof. On
the other hand, the bound on the $L^2$-norm for $f_3$ is only
moderately small and cannot be made to decrease as a function of the complexity
of the decomposition. To get around this issue, we use a sequence of
two decompositions, and make the norm of $f_3$
decrease as a function of the size of the first decomposition.
We hope that these iterated decomposition theorems (proved in a
prequel \cite{BFL12} to this paper) are of independent interest.

\subsubsection{Degree-Structural Properties}
Next, we give an overview of our proof of \cref{thm:degstruct}. For the sake of concreteness, let us focus on a
particular degree-structural property, say, the property $\cP$ of having a
square root, as defined in \cref{sec:introdegstruct}. To show that $\cP$ is locally characterized, we find a
constant $K = K(\cP)$ such that if a function $F: \F^n \to \F$ does
not have a square root, then there must exist a subspace $H$ of dimension $K$
such that $F$ restricted to $H$ also does not have a square root.

So, suppose we are given a function $F: \F^n \to \F$ such that $n \geq
K$ and every hyperplane restriction has a square root of degree $\leq d/2$. For large enough
constant $K$, this automatically implies $\deg(F) \leq d$.
We first {\em  regularize} $f$, meaning that we find polynomials $P_1,
\dots, P_C$ of degree $\leq d$ such that $P_1, \dots, P_C$ are of high
rank and $F = \Gamma(P_1, \dots, P_C)$ for some function
$\Gamma$. Note that here, just as in the proof of the testability
result, we need to allow $P_1, \dots, P_C$ to be non-classical
polynomials. Now, for some $i$ such that $F|_{x_i=0}$ has a square
root, let $P_1', \dots, P_C'$ be the restrictions of $P_1, \dots, P_C$
to $x_i = 0$. So, $\Gamma(P_1',\dots, P_C') = G^2$ for some polynomial
$G$. The polynomials $P_1',\dots, P_C'$ can be shown to be of high rank also. This implies that we can
extend the collection of polynomials $P_1',\dots, P_C'$ to
$P_1',\dots, P_C', Q_1, \dots
Q_D$ such that the new collection is also of high rank and $G
= \Delta(P_1',\dots, P_C', Q_1,\dots, Q_D)$ for some function
$\Delta$. Hence
$$\Gamma(P_1', \dots, P_C') = (\Delta(P_1',\dots, P_C', Q_1, \dots, Q_D))^2.$$
Because of the high rank of the collection $\set{P_1',\dots, P_C',
  Q_1, \dots, Q_D}$, the equidistribution result described in the last
section allows us to conclude that in fact:
$$\Gamma(x_1, \dots, x_C) = (\Delta(x_1, \dots, x_C, y_1, \dots
y_D))^2$$
for all $x_1,\dots, x_C, y_1, \dots, y_D$ in the ranges of $P_1',\dots,
P_C', Q_1, \dots, Q_D$, respectively. Therefore, if we set $\tilde{G} =
\Delta(P_1, \dots, P_C, 0, \dots, 0)$, then $F = \tilde{G}^2$. It is
immediate\footnote{For other degree-structural properties, the degree
  bound may not be immediate at this last step, and we need to argue it separately,
  again using the equidistribution result for high-rank non-classical
  polynomials.} that $\deg(\tilde{G})\leq d/2$, and so, $F$ has a
square root.

It is curious that our proof of \cref{thm:degstruct}, which is
entirely about classical polynomials, requires the use of
non-classical polynomials. Also, as we mentioned earlier, there are no
effective bounds on $K(\cP)$ that arise from our argument. It would be
interesting to obtain better bounds (both upper and lower)
for the locality of degree-structural properties.

\subsection{Comparison with Previous Work}\label{sec:past}
This work is part of, and culminates  a sequence of works investigating the
relationship between affine-invariance and testability. As
described, Kaufman and Sudan \cite{KS08} initiated the
program. Subsequently, Bhattacharyya, Chen, Sudan, and Xie \cite{BCSX09}
investigated {\em monotone} linear-invariant properties of functions
$f:\F_2^n \to \bits$, where a property $\cP$ is monotone if it
satisfies the condition that for any function $g \in \cP$, modifying
$g$ by changing some outputs from $1$ to $0$ does not make it violate
$\cP$. Kr\'al, Serra and Vena \cite{KSV12} and, independently,
Shapira \cite{Sha09} showed testability for any monotone linear-invariant
property characterized by a finite number of linear constraints
(of arbitrary complexity).
For general non-monotone properties, Bhattacharyya, Grigorescu, and
Shapira proved in \cite{BGS10} that affine-invariant properties of
functions in $\{\F_2^n \to \bits\}$ are testable if the complexity of
the property is $1$. Earlier this year, Bhattacharrya, Fischer and
Lovett in \cite{BFL12} generalized \cite{BGS10} to show that
affine-invariant properties of complexity $< p$ are testable.  In this paper, we only have to
restrict the complexity to be bounded, but the bound can be
independent of $p$.

In terms of techniques, the general framework of the proof for
testability here is very much the same as in \cite{BGS10} or
\cite{BFL12}. However, the main difference here is that we work with
collections of non-classical polynomials, rather than classical
ones. Because the degrees of non-classical polynomials can change when
multiplied by constants, the notions of rank and regularity are much
more subtle. We need to establish a new version of a ``polynomial regularity lemma''
which allows us to decompose a given polynomial collection into a high rank
collection of non-classical polynomials. Also, as discussed earlier,
we establish a new equidistribution theorem for non-classical
polynomials. We expect that these results will be of independent
interest.

\ignore{
Higher-order Fourier analysis began with the work of Gowers
\cite{Gow98} and parallel ergodic-theoretic work by Host and Kra
\cite{HK05}. Applications to analytic number theory inspired much more
study by Gowers, Green, Tao, Wolf, and Ziegler among others. A book in
preparation by Tao \cite{Tao11} surveys the current theory of
higher-order Fourier analysis. Our work in this paper relies on
decomposition theorems over finite fields of the type first explicitly
described by Green in \cite{Gre07}. We also heavily use decomposition
results by Hatami and Lovett \cite{HL11}, as described in the
previous section. }

At a high level, the argument to prove our main theorem mirrors ideas
used in a sequence of works \cite{AFKS, AS08a, AS08, graphestim,
  AFNS06, BCLSSV06} to characterize the testable graph properties.  In
particular, the technique of simultaneously decomposing the domain
into a coarse partition and a fine partition with very strong
regularity properties is due to \cite{AFKS}, and the compactness
argument used to handle infinitely many constraints is due to
\cite{AS08a}.

\subsection{Organization}
In \cref{sec:prelude}, we assemble all the technical components and
establish some basic notions such as non-classical polynomials, rank
and regularity. In \cref{sec:equid}, we show the equidistribution
result for non-classical polynomials. In \cref{sec:degstruct}, we use
the results established thus far to prove
\cref{thm:degstruct}. \cref{sec:testing} is devoted to proving
\cref{thm:main3}.

}
\section{Preparation}\label{sec:prelude}

\subsection{Notation}
For integers $a,b$, we let $[a]$ denote the set $\set{1,2,\dots,a}$
and $[a,b]$ denote the set $\set{a, a+1, \dots, b}$. For a set $S$, $\cP(S)$ denotes the
power set of $S$,

Fix a prime field $\F = \F_p$ for a prime $p \geq 2$. As we defined earlier $|\cdot|$
denotes the standard map from $\F$ to $\set{0,1,\dots,p-1} \subset
\Z$.

We use the shorthand $x = a \pm \eps$ to mean $a-\eps \leq x \leq a +
\eps$.

\subsection{Locality}
In the context of affine-invariant properties, we can define the
notion of local characterization in a more algebraic way than we did
in the introduction.  Recall that a {\em hyperplane} is an affine subspace of codimension $1$.

\begin{definition}[Locally characterized properties]
An affine-invariant property $\mcl{P} \subset \{\F^n \to [R] : n \geq 0\}$ is said to
be  {\em locally characterized} if
both of the following hold:
\begin{itemize}
\item
For every function $f: \F^n \to [R]$ in $\mcl{P}$ and every
hyperplane $A \leq \F^n$, $f|_A \in \mcl{P}$.
\item
There exists a constant $K \geq 1$ such that if $f: \F^n \to [R]$ does
not belong to $\mcl{P}$ and $n > K$, then there exists a
hyperplane $B
\leq \F^n$ such that $f|_B \not \in \mcl{P}$.
\end{itemize}

The constant $K$ is said to be the {\em locality} of $\mcl{P}$.
\end{definition}

The following observation shows that an affine-invariant property
is locally characterized if and only if it can be described using a
bounded number of induced affine constraints from the previous
section, and hence, is locally characterized in the sense of the
introduction.
\begin{lemma}
If $\mcl{P} \subseteq \{\F^n \to [R]: n \ge 0 \}$ is a locally characterized
affine-invariant property with locality $K$, then $\mcl{P}$ is
equivalent to $\cA$-freeness, where $\cA$ is a finite collection of
induced affine constraints, with each constraint of size $p^K$ on
$K+1$ variables. On the other hand, if $\mcl{P}$ is equivalent to
$\cA$-freeness, where $\cA$ is a collection of induced affine
constraints with each constraint on $\leq K+1$ variables, then
$\mcl{P}$ has locality at most $K$.
\end{lemma}

Finally, we also make formal note of the observation in the
introduction that if a property is testable, then it must be locally
characterized.
\begin{remark}
If $K$ is a fixed integer and $\mcl{P} \subset \set{\F^n \to [R]}$ is
an affine-invariant property that is testable with $K$ queries, then
$\mcl{P}$ is a locally characterized property with locality $K$.
\end{remark}
So, we can view our main result as a  converse statement.

\subsection{Derivatives and Polynomials}

\begin{definition}[Multiplicative Derivative]\label{multderiv}
Given a function $f: \F^n \to \C$ and an element $h \in \F^n$, define
the {\em multiplicative derivative in direction $h$} of $f$ to be the
function $\Delta_hf: \F^n \to \C$ satisfying $\Delta_hf(x) =
f(x+h)\overline{f(x)}$ for all $x \in \F^n$.
\end{definition}
The {Gowers norm} of order $d$ for a function $f$ is the
expected multiplicative derivative of $f$ in $d$ random directions at
a random point.
\begin{definition}[Gowers norm]\label{gowers}
Given a function $f: \F^n \to \C$ and an integer $d \geq 1$, the {\em
  Gowers norm of order $d$}  for $f$ is given by
$$\|f\|_{U^d} = \left|\E_{h_1,\dots,h_d,x \in \F^n} \left[(\Delta_{h_1}
\Delta_{h_2} \cdots \Delta_{h_d}f)(x)\right]\right|^{1/2^d}.$$
\end{definition}
Note that as $\|f\|_{U^1}= |\Ex{f}|$ the Gowers norm of order $1$ is only a semi-norm. However for $d>1$, it is not difficult to show that
$\|\cdot\|_{U^d}$ is indeed a norm.

If $f = e^{2\pi i P/p}$ where $P: \F^n \to \F$ is a polynomial
of degree $< d$, then $\|f\|_{U^d} = 1$. If $d < p$ and $\|f\|_\infty \le 1$, then in fact, the
converse holds, meaning that any  function $f: \F^n \to \C$ satisfying
$\|f\|_\infty \le 1$ and $\|f\|_{U^d} = 1$ is of this form. But when $d \geq p$, the converse
is no longer true. In order to characterize functions $f : \F^n \to
 \C$ with $\|f\|_\infty \le 1$ and $\|f\|_{U^d}=1$, we define the notion of {\em
  non-classical polynomials}.

Non-classical polynomials might not be necessarily $\F$-valued. We need to
introduce some notation.
Let $\T$ denote the circle group $\R/\Z$. This is an abelian group
with group operation denoted $+$. For an integer $k \geq 0$, let $\U_k$
denote $\frac{1}{p^k} \Z/\Z$, a subgroup of $\T$.
Let $\iota: \F \to \U_1$ be the injection $x \mapsto \frac{|x|}{p} \mod 1$, where $|x|$ is the standard map from $\F$ to
$\set{0,1,\dots,p-1}$. Let $\msf{e}: \T \to \C$ denote the character
$\expo{x} = e^{2\pi i x}$.
\begin{definition}[Additive Derivative]\label{addderiv}
Given a function\footnote{We try to adhere to the following convention: upper-case letters (e.g. $F$ and
  $P$) to denote functions mapping from $\F^n$ to $\T$ or to $\F$,
  lower-case   letters (e.g. $f$ and $g$) to denote functions mapping
  from $\F^n$ to $\C$, and upper-case Greek letters (e.g. $\Gamma$ and
$\Sigma$) to denote functions mapping  $\T^C$ to $\T$. By  abuse of notation, we sometimes conflate $\F$ and $\iota(\F)$.} $P:
\F^n \to \T$ and an element $h \in \F^n$, define
the {\em additive derivative in direction $h$} of $f$ to be the
function $D_hP: \F^n \to \T$ satisfying $D_hP(x) = P(x+h) - P(x)$
for all $x \in \F^n$.
\end{definition}
\begin{definition}[Non-classical polynomials]\label{poly}
For an integer $d \geq 0$, a function $P: \F^n \to \T$ is said to be a
{\em non-classical polynomial of degree $\leq d$} (or simply a
{\em polynomial of degree $\leq d$}) if for all $h_1,
\dots, h_{d+1}, x \in \F^n$, it holds that
\begin{equation}\label{eqn:poly}
(D_{h_1}\cdots D_{h_{d+1}} P)(x) = 0.
\end{equation}
The {\em degree} of $P$ is the smallest $d$ for which the above holds.
A function $P : \F^n \to \T$ is said to be a {\em classical polynomial of degree
$\leq d$} if it is a non-classical polynomial of degree $\leq d$
whose image is contained in $\iota(\F)$.
\end{definition}

It is a direct consequence that a function $f :
\F^n \to \C$ with $\|f\|_\infty \leq 1$ satisfies $\|f\|_{U^{d+1}} =
1$ if and only if $f = \expo{P}$ for a
(non-classical) polynomial $P: \F^n \to \T$ of degree $\leq d$.

The following lemma of Tao and Ziegler shows that a classical
polynomial $P$ of degree $d$ must always be of the form $\iota \circ
Q$, where $Q : \F^n \to \F$ is a polynomial (in the usual sense) of
degree $d$. It also characterizes the structure of non-classical
polynomials.
\begin{lemma}[Part of Lemma 1.7 in \cite{TZ11}]\label{struct}
Let $d \geq 1$ be an integer.
\begin{enumerate}
\item[(i)]
A function $P: \F^n \to \T$ is a polynomial of degree $\leq d+1$ if and only if
$D_hP$ is a polynomial of degree $\leq d$ for all $h \in \F^n$.

\item[(ii)]
A function $P: \F^n \to \T$ is a classical polynomial of degree $\leq
d$ if $P = \iota \circ Q$, where $Q: \F^n \to \F$ has a representation
of the form
$$Q(x_1,\dots,x_n) = \sum_{0\leq d_1, \dots, d_n < p:\atop {\sum_i d_i
    \leq d}}c_{d_1,\dots,d_n}x_1^{d_1}\cdots x_n^{d_n},
$$
for a unique choice of coefficients $c_{d_1,\dots,d_n} \in \F$.

\item[(iii)]
A function $P: \F^n \to \T$ is a polynomial of degree $\leq d$ if and
only if $P$ can be represented as
$$P(x_1,\dots,x_n) = \alpha + \sum_{0\leq d_1,\dots,d_n< p; k \geq 0:
  \atop {0 < \sum_i d_i \leq d - k(p-1)}} \frac{ c_{d_1,\dots, d_n,
  k} |x_1|^{d_1}\cdots |x_n|^{d_n}}{p^{k+1}} \mod 1,
$$
for a unique choice of $c_{d_1,\dots,d_n,k} \in \set{0,1,\dots,p-1}$
and $\alpha \in \T$.  The element $\alpha$ is called the {\em
  shift} of $P$, and the largest integer $k$ such that there
exist $d_1,\dots,d_n$ for which $c_{d_1,\dots,d_n,k} \neq 0$ is called
the {\em depth} of $P$. Classical polynomials correspond to
polynomials with $0$ shift and $0$ depth.

\item[(iv)]
If $P: \F^n \to \T$ is a polynomial of  depth $k$, then it takes
values in a coset of the subgroup $\U_{k+1}$. In particular, a
polynomial of degree $\leq d$ takes on at most
$p^{\lfloor \frac{d-1}{p-1}\rfloor +1}$ distinct values.
\end{enumerate}
\end{lemma}

Note that \cref{struct}~(iii) immediately implies the following important
observation\footnote{Recall that $\T$ is an additive
  group. If $n \in  \Z$  and $x \in \T$, then $nx$ is shorthand for
  $\underbrace{x + \cdots   + x}_{n\text{ terms}}$ if $n \geq 0$ and
  $\underbrace{-x - \cdots   - x}_{-n\text{ terms}}$ otherwise.} :
\begin{remark}
If $Q: \F^n \to \T$ is a polynomial of degree $d$ and depth $k$, then
$pQ$ is a polynomial of degree $\max(d-p+1, 0)$ and depth $k-1$. In
other words, if $Q$ is classical, then $pQ$ vanishes, and otherwise,
its degree decreases by $p-1$ and its depth by $1$. Also, if $\lambda
\in [1, p-1]$ is an integer, then $\deg(\lambda Q) = d$ and
$\mrm{depth}(\lambda Q) = k$.
\end{remark}
Also, for convenience of exposition, we will assume throughout this
paper that the shifts of all polynomials are zero. This can be done
without affecting any of the results in this work. Hence, all
polynomials of depth $k$ take values in $\U_{k+1}$.


\subsection{Inverse Theorem}

There is a tight connection between polynomials and  Gowers
norms. In one direction, it is a straightforward consequence of the
monotonicity of the Gowers norm ($\|f\|_{U^d} \leq \|f\|_{U^{d+1}}$)
and invariance of the Gowers norm with respect to modulation by lower
degree polynomials ($\|f\|_{U^{d+1}} = \|f\cdot \expo{P}\|_{U^{d+1}}$ for
polynomials $P$ of degree $\leq d$) that if $f$ is
{\em $\delta$-correlated} with a polynomial $P$ of degree $\leq d$,
meaning $$|\E_{x} f(x) \expo{-P(x)}| \geq
\delta$$ for some $\delta > 0$,
then $$\|f\|_{U^{d+1}} \geq \delta.$$

In the other direction, we have the following ``Inverse
theorem for the Gowers norm''.
\begin{theorem}[Theorem 1.11 of \cite{TZ11}]\label{inverse}
Suppose $\delta > 0$ and $d \geq 1$ is an integer. There exists an
$\eps = \eps_{\ref{inverse}}(\delta,d)$ such that the following
holds. For every function $f: \F^n \to \C$ with $\|f\|_\infty \leq
1$ and $\|f\|_{U^{d+1}} \geq \delta$, there exists a polynomial $P:
\F^n \to \T$ of degree $\leq d$ that is $\eps$-correlated with $f$,
meaning
$$\left|\E_{x \in \F^n} f(x) \expo{-P(x)} \right| \geq \eps.$$
\end{theorem}

\subsection{Rank}
We will often need to study  Gowers norms of exponentials of polynomials. As we describe below if this analytic
quantity is non-negligible, then there is an algebraic explanation for this: it is possible to  decompose the polynomial as a function of a constant number of low-degree polynomials.  To state this rigorously, let us define the notion of {\em rank} of a polynomial.

\begin{definition}[Rank of a polynomial]
Given a polynomial $P : \F^n \to \T$ and an integer $d > 1$, the {\em $d$-rank} of
$P$, denoted $\msf{rank}_d(P)$, is defined to be the smallest integer
$r$ such that there exist polynomials $Q_1,\dots,Q_r:\F^n \to \T$ of
degree $\leq d-1$ and a function $\Gamma: \T^r \to \T$ satisfying
$P(x) = \Gamma(Q_1(x),\dots, Q_r(x))$. If $d=1$, then
$1$-rank is defined to be $\infty$ if $P$ is non-constant and $0$
otherwise.

The {\em rank} of a polynomial $P: \F^n \to \T$ is its $\deg(P)$-rank.
\end{definition}

Note that for integer $\lambda \in [1, p-1]$,  $\mrm{rank}(P) =
\mrm{rank}(\lambda P)$. The following theorem of Tao and Ziegler shows that high rank polynomials  have small
Gowers norms.

\begin{theorem}[Theorem 1.20 of \cite{TZ11}]\label{arank}
For any $\eps > 0$ and integer $d > 0$, there exists an
integer $r = r_{\ref{arank}}(d,\eps)$ such that the following is
true. For any polynomial $P : \F^n \to \T$ of degree $\leq d$, if
$\|\expo{P}\|_{U^d} \geq \eps$, then $\msf{rank}_{d}(P) \leq r$.
\end{theorem}

For future use, we also record here a simple
lemma stating that restrictions of high rank polynomials to
hyperplanes generally preserve degree and high rank.

\begin{lemma}\label{rankrestrict}
Suppose $P: \F^n \to \T$ is a polynomial of degree $d$ and rank
$\geq r$, where $r > p+1$. Let $A$ be a hyperplane in $\F^n$, and
denote by $P'$ the restriction of $P$ to $A$. Then,
$P'$ is a polynomial of degree $d$ and  rank $\geq r-p$, unless $d=1$
and $P$ is constant on $A$.
\end{lemma}
\begin{proof}
For the case $d=1$, we can check directly that either $P'$ is constant
or else, $P'$ is a non-constant degree-$1$ polynomial and so has
rank infinity.

So, assume $d>1$.
By making an affine transformation, we can assume without loss of
generality that $A$ is the hyperplane $\set{x_1=0}$. Let $\pi: \F^n
\to \F^{n-1}$ be the projection to $A$. Let $P'' = P - P' \circ
\pi$. Clearly, $P''$ is zero on $A$. For $x \in \F \setminus \set{0}$,
let $h_x = (x,0,\dots,0) \in \F^n$. Note that $D_{h_x}P''$ is of
degree $\leq d-1$ and that $(D_{h_x}P'')(y) = P''(y+h_x)$ for all $y
\in A$. Hence, for every $x \in \F \setminus \set{0}$,  $P''$ on $h_x
+ A$ agrees with a polynomial $Q_x$ of degree $\leq d-1$. So, for a
function $\Gamma: \T^{p+1} \to \T$, we can write $P =
\Gamma(\iota(x_1), P', Q_1, Q_2,\dots,Q_{p-1})$, where
$\iota(x_1), Q_1,\dots,Q_{p-1}$ are of degree $\leq d-1$.

Now, if $P'$ itself is of degree $d-1$, then ${P}$ is of rank
$\leq p+1 < r$, a contradiction. If $P'$ is of rank $< r-p$, then again
${P}$ is of rank $< r-p + p = r$, a contradiction.
\end{proof}

\subsection{Polynomial factors}

A high-rank polynomial of degree $d$ is, intuitively, a ``generic''
degree-$d$ polynomial. There are no unexpected ways to decompose it
into lower degree polynomials, and the property of high rank is robust
against various operations such as restrictions to hyperplanes, taking
derivatives, multiplying by integers, etc. Next, we will formalize the
notion of a generic collection of polynomials. Intuitively, it should
mean that there are no unexpected algebraic dependencies among the
polynomials. First, we need to set up some notation.

\begin{definition}[Factors]. If $X$ is a finite set then by a \emph{factor} $\cB$ we mean simply a
partition of $X$ into finitely many pieces called \emph{atoms}.
\end{definition}

A function $f:X \to \C$ is called \emph{$\cB$-measurable} if it is constant on atoms of $\cB$. For any function $f : X \to \C$, we may define 
the conditional expectation 
$$\E[f|\cB](x)=\E_{y \in \cB(x)}[f(y)],$$ 
where $\cB(x)$ is the unique atom in $\cB$ that contains $x$. Note that $\E[f|\cB]$ is $\cB$-measurable. 

A finite collection of functions $\phi_1,\ldots,\phi_C$ from $X$ to some other space $Y$ naturally define a factor $\cB=\cB_{\phi_1,\ldots,\phi_C}$ whose atoms are sets of the form $\{x: (\phi_1(x),\ldots,\phi_C(x))= (y_1,\ldots,y_C) \}$ for some $(y_1,\ldots,y_C) \in Y^C$. By an abuse of notation 
we also use $\cB$ to denote the map $x \mapsto (\phi_1(x),\ldots,\phi_C(x))$, thus also identifying the atom containing $x$ with 
$(\phi_1(x),\ldots,\phi_C(x))$. 

\begin{definition}[Polynomial factors]\label{factor}
If $P_1, \dots, P_C:\F^n \to \T$ is  a sequence of polynomials, then the factor $\cB_{P_1,\ldots,P_C}$ is called a {\em polynomial factor}.
\end{definition}

The {\em complexity} of $\cB$, denoted $|\cB|$, is the number of defining polynomials
$C$.  The {\em degree} of $\cB$ is the maximum
degree among its defining polynomials $P_1,\ldots,P_C$. If $P_1,\ldots,P_C$ are of depths $k_1,\ldots,k_C$, respectively,
then $\|\cB\|=\prod_{i=1}^C p^{k_i+1}$ is called the \emph{order} of $\cB$.

Notice that the number of atoms of $\cB$ is bounded by $\|\cB\|$.
\ignore{
We will endow polynomial factors with some
additional structure that will be useful for dealing with non-classical
polynomials.
\begin{definition}[Extended Polynomial factor] \label{extfactor}
Given an integer $d > 0$, an {\em extended polynomial factor of degree
  $d$} is a polynomial factor $\cB$ defined by a sequence of
polynomials $P_1, \dots, P_C: \F^n \to \T$ such that if $P \in \cB$
and $P$ is not classical, then $pP \in \cB$ also.
\ignore{a
tuple $\cB$ of polynomials $(P_{i,j,k})_{i \in
  [d], j \in [C_i], k \in [0,M_{i,j}]}$ where each $C_i$ and each
$M_{i,j}$ is a non-negative integer\mnote{Maybe like [TZ],
  we also want all polynomials of degree $>1$?}. Each $P_{i,j,k} :
\F^n \to \T$ is a polynomial of degree $i + k(p-1)$. Additionally, $pP_{i,j,k}
= P_{i,j,k-1}$ for all $i \in [d], j \in [C_i], k \in [M_{i,j}]$.  In
particular, each $P_{i,j,k}$ is a polynomial of depth $k$ and thus
takes values in $\U_{k_i+1}$. Note that
the total number of polynomials in the factor is $|\cB| = \sum_{i \in
  [d]}\sum_{j \in   [C_{i}]} (M_{i,j}+1)$ and that the total number of
atoms is $\|\cB\| = \prod_{i \in  [d]}\prod_{j \in   [C_{i}]} \prod_{k \in
  [0,M_{i,j}]} p^{k+1}$.}
\end{definition}}
The rank of a factor can now be defined as follows.
\begin{definition}[Rank and Regularity]\label{regular}
A polynomial factor $\cB$ defined by a sequence of
polynomials $P_1,\dots,P_C: \F^n \to \T$ with respective depths $k_1,
\dots, k_C$ is said to have {\em   rank
  $r$} if $r$ is the least integer for which there exist
$(\lambda_1,\dots, \lambda_C) \in \Z^C$ so that $(\lambda_1 \mod p^{k_1 + 1}, \dots, \lambda_C \mod
p^{k_C + 1}) \neq (0, \dots, 0)$
and the polynomial $Q = \sum_{i=1}^C \lambda_i P_i$ satisfies
$\msf{rank}_{d}(Q) \leq r$ where $d = \max_{i} \deg(\lambda_iP_i)$.

Given a polynomial factor $\cB$ and a function $r: \Z_{> 0}
\to \Z_{>   0}$, we say $\cB$ is {\em $r$-regular} if $\cB$ is of rank
larger than $ r(|\cB|)$.
\end{definition}

 Note that since $\lambda$ can be a multiple of $p$, rank measured with
respect to $\deg(\lambda P)$ is not the same as rank measured with
respect to
$\deg(P)$. So, for instance, if $\cB$ is the factor defined by a
single polynomial $P$ of degree $d$ and depth $k$, then
$$\msf{rank}(\cB) = \min\set{\msf{rank}_d(P), \msf{rank}_{d-(p-1)}(pP),
\cdots, \msf{rank}_{d-k(p-1)}(p^kP)}.$$

Regular factors indeed do behave like a generic collection of
polynomials, as we shall establish in a precise sense in
\cref{sec:equid}. Thus, given any factor $\cB$ that is not regular, it
will often be useful to {\em regularize} $\cB$, that is, find a refinement
$\cB'$ of $\cB$ that is regular up to our desires. We distinguish
between two kinds of refinements:
\ignore{
Next, we define the notion of conditional expectation with respect to
a given factor.

\begin{definition}[Expectation over polynomial factor]
Given a factor $\cB$ and a function $f:\F^n\to\bits$, the
{\em expectation} of $f$ over an atom $y \in \T^{|\cB|}$ is the average
$\E_{x:\cB(x)=y}[f(x)]$, which we denote by $\E[f|y]$. The {\em
  conditional expectation} of $f$ over $\cB$, is the real-valued
function over $\F^n$ given by $\E[f|\cB](x)=\E[f|\cB(x)]$. In
particular, it is constant on every atom of the polynomial factor.
\end{definition}

We use decomposition theorems that partition the domain $\F^n$
into finer and finer partitions. We will need to be careful about
distinguishing between two types of refinements.}

\begin{definition}[Semantic and syntactic refinements] \label{refine}
$\cB'$ is called a {\em syntactic refinement} of $\cB$, and
denoted $\cB' \succeq_{syn} \cB$, if the sequence of polynomials
defining $\cB'$ extends that of $\cB$. It is called a {\em
  semantic refinement}, and denoted $\cB' \succeq_{sem} \cB$ if the
induced partition is a combinatorial refinement of the partition
induced by $\cB$. In other words, if for every $x,y\in \F^n$,
$\cB'(x)=\cB'(y)$ implies $\cB(x)=\cB(y)$. 
\end{definition}
\begin{remark}
Clearly, being a syntactic refinement is stronger than being a
semantic refinement. But observe that  if $\cB'$ is a semantic
refinement of $\cB$, then there exists a
syntactic refinement $\cB''$ of $\cB$ that induces the same
partition of $\F^n$, and for which $|\cB''|\leq
|\cB'|+|\cB|$, because we can define $\cB''$ by
just adding the defining polynomials of $\cB$ to those of $\cB'$.
\end{remark}

The following lemma is the workhorse that allows us to construct
regular refinements.
\begin{lemma}[Polynomial Regularity Lemma]\label{factorreg}
Let $r: \Z_{>0} \to \Z_{>0}$ be a non-decreasing function and $d > 0$
be an integer. Then, there is a  function
$C_{\ref{factorreg}}^{(r,d)}: \Z_{>0} \to \Z_{>0}$  such
that the following is true. Suppose $\cB$ is a factor defined by
polynomials $P_1,\dots, P_C : \F^n \to \T$  of degree at most $d$.
Then, there is  an $r$-regular factor $\cB'$ consisting of  polynomials
$Q_1, \dots, Q_{C'}: \F^n \to \T$ of degree $\leq d$ such that $\cB'
\succeq_{sem} \cB$ and  $C' \leq  C_{\ref{factorreg}}^{(r,d)}(C)$.

Moreover, if $\cB$ is itself a refinement of some
$\hat{\cB}$ that has rank $> (r(C') + C')$ and consists of polynomials, then additionally
$\cB'$ will be a syntactic refinement of $\hat{\cB}$.
\end{lemma}
\begin{proof}
We can prove our lemma starting from Lemma 9.6 of \cite{TZ11}. To explain,
let us define the notion of an {\em extended} factor.
We say a polynomial factor $\cB$ is {extended} if for any
polynomial $Q \in \cB$ that is not classical, $pQ \in \cB$ also. Note that an extended factor
defined by polynomials $P_1, \dots, P_C$ is of high rank if for all
tuples  $(\lambda_1, \dots, \lambda_C) \in [0, p-1]^C$,  unless all the
$\lambda_i$'s are zero, $\sum_i\lambda_i P_i$ is of high
$(\max_i\deg(\lambda_i P_i))$-rank.
Tao and Ziegler proved the following:
\begin{lemma}[Lemma 9.6 of \cite{TZ11}]\label{tzreg}
Let $r: \Z_{>0} \to \Z_{>0}$ be a non-decreasing function and $d > 0$
be an integer. Then, there are  functions
$\bar{C}^{(r,d)}: \Z_{>0} \to \Z_{>0}$  and $I^{(d)}:\Z_{>0} \to \Z_{>0}$
such that the following is true. Suppose $\cB$ is an extended polynomial factor
defined by polynomials $P_1,\dots, P_C : \F^n \to \T$ of degree $\leq d$.
Then, there is a subspace $V \leq \F^n$  and an $r$-regular extended
factor $\bar{\cB}$ consisting of polynomials $Q_1, \dots, Q_{\bar{C}}:
V \to \T$ such that $2 \leq \deg(Q_i) \leq d$ for each
$i$, $\bar{\cB}$
semantically refines the factor defined by $P_1|_V, \dots, P_C|_V$,
$\bar C \leq  \bar{C}^{(r,d)}(C)$, and $\dim(V) \geq n-I^{(d)}(\bar C)$.
\end{lemma}

Let $\cB_1$ be the extended factor defined by $\set{p^k P_i \mid 0
  \leq k \leq
  \mrm{depth}(P_i), i \in [C]}$.
Apply \cref{tzreg} to $\cB_1$ in order to obtain a bounded index
subspace $V_1$ and an extended $R_1$-regular
factor $\bar{\cB}_1$ defined by polynomials $Q_1, \dots, Q_{\bar C}: V_1
\to \T$, where $R_1$ is a growth function (growing even faster than $r$)
we specify later on in the
proof and $\bar{C} \leq \bar{C}^{(R_1,d)}(|\cB_1|)$.  For $a \in \F^n/V_1$
and $P \in \cB_1$, define $P^a: V_1 \to
\F^n$ to be $P^a(x) = D_aP(x) = P(a+x) - P(x)$. Each $P^a$
is of degree $\leq d-1$.  Also, since $V_1$ is the intersection of $I
\leq I^{(d)}(\bar{C})$ hyperplanes, we can decide which coset in $\F^n/V_1$
an element $x \in \F^n$ belongs to as a function of  $I \leq
I^{(d)}(\bar C)$ (classical) linear functions $\pi_1, \dots, \pi_I$. Let
${\cB}_2$ be the extended factor obtained by adding to $\bar{\cB}_1$ all the
polynomials $\set{P^a \mid P \in \cB_1, a \in \F^n/V_1}$ and $\pi_1, \dots, \pi_I$. Consider $x \in \F^n$ and let $x=a+y$ where $y \in V_1$, and $a \in \F^n/V_1$.  Since $P(x)=P(y)-P^{-a}(x)$, each polynomial in $\cB_1$ is a function of the polynomials in
${\cB}_2$ over all of $\F^n$, and so ${\cB}_2$ is a semantic refinement of
$\cB_1$ (and a syntactic refinement of $\bar{\cB}_1$). Note that
$|\cB_2| \leq \bar{C} + dCI^{(d)}(\bar C) + I^{(d)}(\bar C) < \bar{C}
+ 2d\bar{C}I^{(d)}(\bar{C})$.

Now, suppose we repeat the steps in the previous paragraph with
${\cB}_2$ taking the place of  $\cB_1$ and a different function $R_2$
taking the place of $R_1$. We specify $R_2$ later, but we will choose
it so that it grows faster than $r$. The new application of
\cref{tzreg} to $\cB_2$ produces an extended factor $\bar{\cB}_2$ that is
$R_2$-regular and a bounded index subspace $V_2$ such that the
polynomials in $\cB_2$ restricted to $V_2$ are measurable with respect to $\bar{\cB}_2$. We argue that $\bar{\cB}_2$ differs from ${\cB}_2$
only by polynomials of degree $\leq d-1$. Suppose ${\cB}_2$ is not
$R_2$-regular to start off with. The function $R_1$ is
chosen so that $\bar{\cB}_1$'s rank, $R_1(|\bar{\cB}_1|) > R_2(|\cB_2|) + |\cB_2|$. This
means that if a linear combination of polynomials in $\cB_2$, $\sum_{S
  \in \cB_2} \lambda_SS$, has rank $\leq R_2(|\cB_2|)$ and $d' = \max_{S:
  \lambda_S \neq 0} \deg(S)$, then there must be an $S \not \in
\bar{\cB}_1$ with $\lambda_S \neq 0$ and degree $d'$, since otherwise the rank condition of
$\bar{\cB}_1$ would be violated. Since all the polynomials in $\cB_2$ which are
not in $\bar{\cB}_1$ have degree $\leq d-1$, we conclude that $d' \leq d-1$. Inspecting
the proof of \cref{tzreg} in \cite{TZ11} shows that this means
$\bar{\cB}_2$  consists of the
polynomials of $\bar{\cB}_1$ along with other polynomials of degree $\leq
d-1$.  In the same way as in the previous paragraph, we  obtain an
extended factor $\cB_3 \succeq_{syn}\bar{\cB}_2$, so that $\cB_3$ is a semantic refinement of $\cB_2$ over
all of $\F^n$. Note that since all the polynomials of  $\bar{\cB}_1$  are already in $\bar{\cB}_2$, we  only need to add  $\set{P^a \mid P \in \cB_2\setminus \bar{\cB}_1, a \in \F^n/V_2}$, together with some linear functions. All these polynomials have degree at most $\leq d-2$.

We keep repeating this process to obtain a sequence of extended
factors $\cB_1, \cB_2, \cB_3, \dots$ and $\bar{\cB}_1, \bar{\cB}_2,
\bar{\cB}_3, \dots$. Each $\cB_{i+1}$ semantically refines $\cB_i$ and
syntactically refines $\bar{\cB}_i$. The process stops at step $i$ if
$\cB_i$ becomes $R_i$-regular,  where the sequence of growth
functions $R_i$ satisfies $R_i(m) >
R_{i+1}(m+2dmI^{(d)}(m))+m+2dmI^{(d)}(m)$ and $R_d(m) = r(m)$.
The functions $R_i$ are chosen so that $R_i(|\bar{\cB}_i|) >
R_{i+1}(|\cB_{i+1}|) + |\cB_{i+1}|$, and therefore, by the above argument, $\cB_{i+1}$
differs from $\cB_i$ by polynomials of degree $\leq d-i$. So, we must
stop after obtaining $\cB_d$ in the sequence. Also, since each $R_i$
grows faster than $r$, note that $R_i$-regularity for any $i \in [d]$
implies $r$-regularity. So, it must be
that some $\cB_i$ for $i \leq d$ already becomes  $r$-regular.

Given an extended factor $\cB''$ of rank $>r$, we can get a (standard)
factor $\cB'$ of rank $>r$ by letting $\cB'$ be defined by the
smallest subset of polynomials $S$ such that $\set{p^i P \mid P \in S,
  i \in \Z_{\geq 0}} \supseteq \cB''$. The last statement of the lemma
follows from the same considerations as used above to argue that
$\bar{\cB}_i$ syntactically refines $\cB_i$.
\end{proof}


\section{Equidistribution of Regular Factors}
\label{sec:equid}

In this section, we make precise the intuition that a high-rank
collection of polynomials often behaves like a collection of
independent random variables. The key technical tool is the connection
between the combinatorial notion of rank and the analytic
notion of bias, given in \cref{arank}. A weaker statement, that was
established earlier by Kaufman and Lovett\footnote{Kaufman and Lovett
  proved \cref{rankreg} for classical polynomials. But their proof
  also works for non-classical ones without modification.} and used by
Tao and Ziegler in their proof of \cref{arank}, is the
following.

\begin{theorem}[Theorem 4 of \cite{KL08}]\label{rankreg}
For any $\eps > 0$ and integer $d > 0$, there exists
$r = r_{\ref{rankreg}}(d,\eps)$ such that the following is true.
If $P: \F^n \to \T$ is a degree-$d$ polynomial with rank greater than $r$,
then $|\E_{x}[\expo{P(x)}]| < \eps$.
\end{theorem}
\begin{proof}
Given \cref{arank}, this follows directly from easy fact that $\left|\Ex{f}\right| \le \|f\|_{U^{d}}$ for every $
d \ge 2$, and every $f:\F^n \to \C$.
\end{proof}

Using a standard observation that relates the bias of a function to
its distribution on its range, we can conclude the following.

\begin{lemma}[Size of atoms]\label{atomsize}
Given $\eps > 0$, let $\cB$ be a polynomial factor of
degree $d > 0$,  complexity $C$, and rank
$r_{\ref{rankreg}}(d,\eps)$,   defined by a tuple of
polynomials $P_1, \dots, P_C: \F^n
\to \T$ having respective depths $k_1, \dots, k_C$.
Suppose $b = (b_1, \dots, b_C) \in \U_{k_1+1} \times \cdots \times \U_{k_C+1}$. Then
$$
\Pr_{x}[\cB(x) = b] =  \frac{1}{\|\cB\|} \pm \eps.
$$
In particular, for  $\eps < \frac{1}{\|\cB\|} $, $\cB(x)$ attains every possible value in its range and thus has
$\|\cB\|$ atoms.
\end{lemma}
\begin{proof}
\begin{align*}
\Pr_{x}[\cB(x) = b]
&= \E_{x} \left[ \prod_{i}\frac{1}{p^{k_i+1}} \sum_{
    \lambda_{i}=0}^{p^{k_i+1}-1} \expo{\lambda_{i} (P_{i}(x) - b_{i})}\right]\\
&= \prod_{i}p^{-(k_i+1)} \cdot \sum_{(\lambda_1, \dots, \lambda_C)
\atop \in \prod_{i} [0,p^{k_i + 1}-1]}
\E_{x}\left[\expo{\sum_{i} \lambda_{i}(P_{i}(x) - b_{i})}\right]\\
&= \prod_{i}p^{-(k_i+1)}  \cdot \left(1 \pm \eps \prod_{i}p^{k_i+1}\right)= \frac{1}{\|\cB\|} \pm \eps.
\end{align*}
The first equality uses the fact that $P_{i}(x) - b_{i}$ is
in $\U_{k_i+1}$ and that for any nonzero $x \in \U_{k_i+1}$,
$\sum_{\lambda   = 0}^{p^{k+1}-1} \expo{\lambda  x} = 0$.  The
third equality uses \cref{rankreg} and the fact that unless every
$\lambda_{i} = 0$, the polynomial $\sum_{i} \lambda_{i}(P_{i}(x) -
b_{i})$ has rank at least $r_{\ref{rankreg}}(d,\eps)$.
\end{proof}

For our applications, we need to not only understand the distribution
of $\cB(x) = (P_{i}(x))$ but also, more generally,
$(P_{i}(L_j(x)))$ for a given sequence of linear forms
$L_1,\dots,  L_m : (\F^n)^\ell \to \F^n$.  To this end, we first show the
following dichotomy theorem.

\begin{theorem}[Near orthogonality]\label{dich}
Given $\eps > 0$, suppose
 $\cB = (P_1, \dots, P_C)$ is a polynomial factor of
degree $d > 0$ and rank $> r_{\ref{arank}}(d,\eps)$, $A
=(L_1,\dots,L_m)$ is an affine constraint
on $\ell$ variables, and $\Lambda$ is a tuple of
integers $(\lambda_{i,j})_{i \in [C], j \in [m]}$. Define
$$P_{A,\cB, \Lambda}(x_1,\dots,x_\ell) = \sum_{i \in [C], j \in [m]} \lambda_{i,j}
P_{i}(L_j(x_1,\dots,x_\ell)).$$
Then, one of the two statements below is true.
\begin{itemize}
\item
For every $i \in [C]$, it holds that $\sum_{j \in   [m]}
\lambda_{i,j} Q_{i}(L_j(\cdot)) \equiv 0$ for all
polynomials $Q_i: \F^n \to \T$ with the same degree and depth as
$P_i$. Clearly, $P_{A,\cB, \Lambda} \equiv 0$ in this case.
\item
$P_{A,\cB, \Lambda}$ is non-constant. Moreover, $|\E_{x_1,\dots,x_\ell}[\expo{P_{A,\cB, \Lambda}(x_1, \dots,
  x_\ell)}]| < \eps$.
\end{itemize}
\end{theorem}
\begin{proof}
If $\lambda_{i,j} \neq 0$, then $\lambda_{i,j}P_{i}$
can be assumed to be non-constant, since otherwise, we can set
$\lambda_{i,j}$ to $0$.  Let the depths of $P_1, \dots, P_C$ be
$k_1, \dots, k_C$ respectively. For each $j \in [m]$, we let
$(w_{j,1}, \dots, w_{j,\ell})$ denote the vector corresponding to the
affine form $L_j$; recall that $w_{j,1} = 1$. For any affine form
$L_j$, let $|L_j|$, its {\em weight}, denote the sum $\sum_{t=2}^\ell
|w_{j, t}|$.

 For each $i$, perform the following step independently.
If there exists a $j$ such that $|L_j| > \deg(\lambda_{i,j}
P_{i})$ and $\lambda_{i,j} \neq 0$, then use
\cref{eqn:poly} to replace $\lambda_{i,   j}
P_{i}(L_j(\cdot))$ by a linear combination over $\Z$ of
$P_{i}(L_{j'}(\cdot))$ with $L_{j'} \preceq
L_j$, and repeat until no such $j$ exists. Here we use the assumption in part (ii) of
\cref{defaffine} that such $L_{j'} \in A$. At the end of this process, we obtain a
new tuple of coefficients $\Lambda' = (\lambda'_{i,j})$; we can
assume that each $\lambda'_{i,j} \in [0, p^{k_i+1}-1]$ after
quotienting with $p^{k_i+1}\Z$.

 If all the $\lambda'_{i,j}$ are zero, then  for the original
 coefficients $(\lambda_{i,j})$ also,   $\sum_{j \in  [m]}
 \lambda_{i,j} P_{i}(L_j(\cdot))$ is identically zero for every $i$
 individually. Indeed, $\sum_{j \in [m]} \lambda_{i,j}Q_i(L_j(\cdot))$
 is zero for any $Q_i$ with the same degree and depth as  $P_i$
 because the above transformation from $\Lambda$ to $\Lambda'$
 only  depended upon the degree and depth of $P_i$.

Otherwise, $\Lambda'$ does not consist of all zeroes,
 and for every nonzero $\lambda'_{i,j}$, we have $|L_j| \leq
 \deg(\lambda'_{i,j} P_{i})$.  In this case we show that
 $|\E[\expo{P_{A,\cB, \Lambda'}(x_1, \dots,x_\ell)}]| < \eps$.  At a high
 level, our goal is to express the bias of $P_{A,\cB, \Lambda'}$ in
 terms of the Gowers norm of a linear combination of $P_i$'s and then
 use \cref{arank}.

Suppose without loss of generality that the form $L_1$ satisfies:
\begin{itemize}
\item[(i)] $\lambda'_{i,1} \neq 0$ for some $i \in [C]$.
\item[(ii)] $L_1$ is maximal in the sense that for every $j \neq 1$,
  either $\lambda'_{i,j} = 0$ for all $i \in [C]$ or it is the
  case that $|w_{j,t}| < |w_{1,t}|$ for some  $t \in [\ell]$.
\end{itemize}

We want to ``derive'' $P_{A,\cB, \Lambda'}$ until we kill all
$P_i(L_j(\cdot))$ terms for $j > 1$. Given a vector $\bm{\alpha} =
(\alpha_1, \dots, \alpha_\ell) \in \F^\ell$, an element $y \in \F^n$,  and a
function $P : (\F^n)^\ell \to \T$, let us define
$$\bar{D}_{\bm{\alpha},y} P( x_1, \dots, x_\ell) = P(x_1 +
\alpha_1y, \dots, x_\ell + \alpha_\ell y) - P( x_1,
 \dots, x_\ell).
$$
Note that
\begin{align*}
\bar{D}_{\bm{\alpha},y} (P_i\circ L_j)(x_1, \dots, x_\ell)
&= P_i(L_j(x_1,
\dots, x_\ell) + L_j(\bm{\alpha}) y) - P_i(L_j(x_1, \dots,
x_\ell))\\
&= (D_{L_j(\bm{\alpha}) y}P_i)(L_j(x_1,\dots,x_\ell)).
\end{align*}
Thus, if $L_j(\bm{\alpha}) = \langle L_j, \bm{\alpha}\rangle = 0$, then
$\bar{D}_{\bm{\alpha},y} P_i\circ L_j \equiv 0$ for all choices of $y$.

Set $\Delta = |L_1|=\sum_{i=2}^\ell w_{1,i}$, and let $\bm{\alpha}_1,
\dots, \bm{\alpha}_\Delta \in \F^\ell$ be the set of all  vectors of the form $(-w,0,\ldots,0,1,0,\ldots,0)$ where $1$ is in the $i$th coordinate for $i \in [2,\ell]$ and $0 \le w \le |w_{1,i}|-1$ is an integer. Note that that $\langle
L_1,\bm{\alpha}_k \rangle \neq 0$ for all $k \in
[\Delta]$, but for any $j > 1$, by maximality of $L_1$, there exists some $k \in [\Delta]$
such that $\langle L_j,\bm{\alpha}_k \rangle = 0$. Consequently,
\begin{align*}
\bar{D}_{\bm{\alpha}_\Delta, y_\Delta}\cdots
\bar{D}_{\bm{\alpha}_1, y_1} P_{A,\cB,
  \Lambda'}(x_1,\dots, x_\ell)
&= \left(\bar{D}_{\bm{\alpha}_\Delta, y_\Delta}\cdots
\bar{D}_{\bm{\alpha}_1, y_1}\sum_{i = 1}^C \lambda'_{i,1} P_{i}\circ
L_{1}\right)(x_1,\dots, x_\ell) \\
&= (D_{\langle L_1, \bm{\alpha}_\Delta\rangle y_\Delta} \cdots
D_{\langle L_1, \bm{\alpha}_1\rangle y_1} \sum_{i=1}^C
\lambda'_{i,1}P_i)(L_1(x_1,\dots,x_\ell)).
\end{align*}
Therefore
\begin{align}\label{clinch}
\E_{y_1,\dots, y_\Delta, \atop x_1, \dots,
  x_\ell}[\expo{(\bar{D}_{\bm{\alpha}_\Delta, y_\Delta}\cdots
\bar{D}_{\bm{\alpha}_1, y_1} P_{A,\cB,
  \Lambda'})( x_1,\dots, x_\ell)}]
= \left\|\sum_{i=1}^C
\lambda'_{i,1}P_i\right\|_{U^\Delta}^{2^\Delta}.
\end{align}
On the other hand we have the following claim.
\begin{claim}\label{clm:cs}
\begin{align*}
\E_{y_1,\dots, y_\Delta,\atop x_1, \dots,
  x_\ell}[\expo{(\bar{D}_{\bm{\alpha}_\Delta, y_\Delta}\cdots
\bar{D}_{\bm{\alpha}_1, y_1} P_{A,\cB,
  \Lambda'})(x_1,\dots, x_\ell)}]
\geq \left(\left|\E_{x_1,\dots,x_\ell} \expo{P_{A,\cB,
        \Lambda'}(x_1,\dots,x_\ell)} \right|\right)^{2^{\Delta}}.
\end{align*}
\end{claim}
\begin{proof}
It suffices to show that for any function $P(x_1,\dots,x_\ell)$ and
nonzero $\bm{\alpha} \in \F^\ell$,
$$\left|\E_{y, x_1, \dots, x_\ell}[\expo{(\bar{D}_{\bm{\alpha},
    y}P)(x_1,\dots,x_\ell)}]\right| \geq
\left|\E_{x_1, \dots, x_\ell}[\expo{P(x_1,\dots,x_\ell)}]\right|^2.
$$
Recall that $(\bar{D}_{\bm{\alpha},
    y}P)(x_1,\dots,x_\ell) = P(x_1 + \alpha_1 y, \dots, x_\ell + \alpha_\ell y)
  - P(x_1, \dots, x_\ell)$. Without loss of generality, suppose $\alpha_1
  \neq 0$. We make a change of coordinates  so that
  $\bm{\alpha}$ can be assumed to be $(1,0,\dots, 0)$. More precisely,
  define $P': (\F^n)^\ell \to \T$ as
$$P'(x_1, \dots, x_\ell) = P\left(x_1,
  \frac{x_2 + \alpha_2 x_1}{\alpha_1}, \frac{x_3 + \alpha_3 x_1}{\alpha_1}, \dots,
  \frac{x_\ell + \alpha_\ell x_1}{\alpha_1}\right),$$
so that $P(x_1, \dots, x_\ell) = P'(x_1,  \alpha_1 x_2 - \alpha_2 x_1, \alpha_1 x_3 - \alpha_3x_1, \dots,
  \alpha_1 x_\ell - \alpha_\ell x_1)$, and thus $(\bar{D}_{\bm{\alpha},
    y}P)(x_1,\dots,x_\ell) = P'(x_1 + \alpha_1 y, \alpha_1 x_2 - \alpha_2
  x_1, \dots, \alpha_1 x_\ell -\alpha_\ell x_1) - P'(x_1, \alpha_1 x_2 - \alpha_2
  x_1, \dots, \alpha_1 x_\ell -\alpha_\ell x_1)$. Therefore
\begin{align*}
&\left|\E_{y, x_1, \dots, x_\ell}[\expo{(\bar{D}_{\bm{\alpha},
    y}P)(x_1,\dots,x_\ell)}]\right|\\
&= \left|\E_{y, x_1, \dots, x_\ell}[\expo{P'(x_1 + \alpha_1 y, \alpha_1 x_2 - \alpha_2
  x_1, \dots, \alpha_1 x_\ell -\alpha_\ell x_1) - P'(x_1, \alpha_1 x_2 - \alpha_2
  x_1, \dots, \alpha_1 x_\ell -\alpha_\ell x_1)}]\right|\\
&= \left|\E_{y, x_1, \dots, x_\ell}[\expo{P'(x_1 + \alpha_1 y, x_2,
    \dots, x_\ell) -  P'(x_1, x_2, \dots, x_\ell)}]\right|= \E_{x_2,\dots,x_\ell}\left|\E_{x_1}[\expo{P'(x_1, x_2, \dots,
    x_\ell)}]\right|^2\\
&\geq \left|\E_{x_1, x_2, \dots, x_\ell} [\expo{P'(x_1, x_2, \dots,
    x_\ell)}]\right|^2
= \left|\E_{x_1, x_2, \dots, x_\ell} [\expo{P(x_1, x_2, \dots,
    x_\ell)}]\right|^2.
\end{align*}

\end{proof}

Therefore, combining \cref{clinch} with \cref{clm:cs}, we get:
$$
\left\| \expo{\sum_{i \in
  [C]} \lambda'_{i,1} P_{i}(x)} \right\|_{U^{\Delta}}
\geq \left|\E_{x_1,\dots,x_\ell} \expo{P_{A,\cB,
        \Lambda'}(x_1,\dots,x_\ell)} \right|.
$$

Suppose $|\E_{x_1,\dots,x_\ell} \expo{P_{A,\cB,
        \Lambda'}(x_1,\dots,x_\ell)}| \geq \eps$. Then, by the above
    inequality and \cref{arank}, we get that $\sum_{P_{i} \in
  \cB} \lambda'_{i,1} P_{i}(x)$ is a function of $r =
r_{\ref{arank}}(d, \eps)$ polynomials of degree $\Delta - 1$. But
recall that if $\lambda'_{i,1} \neq 0$, then $\deg(\lambda'_{i,1}P_{i}) \geq |L_1| =
\Delta$. Also, there exists a nonzero $\lambda'_{i,1}$. This is a
contradiction to our assumption that the factor $\cB$ is of rank $>
r_{\ref{arank}}(d, \eps)$.
\end{proof}

\begin{remark}\label{rmk:small}
The proof of \cref{dich} also shows the following. Suppose, in the
setting of \cref{dich}, that
for every $P_{i} \in \cB$ and $L_j \in A$, either $|L_j|
\leq \deg(\lambda_{i,j}P_{i})$ or $\lambda_{i,j} = 0$.
Then, unless every $\lambda_{i,j} = 0 \pmod {p^{k_i+1}}$, we have that
$P_{A,\cB,\Lambda}$ is non-constant and
$|\E[\expo{P_{A,\cB,\Lambda}(x_1, \dots, x_\ell)}]| < \eps$. The only
modification needed to the above proof is that the transformation from $\Lambda$ to
$\Lambda'$ can be omitted.
\end{remark}

To show equidistribution of $(P_i(L_j(x_1, \dots, x_\ell))$, we can use
\cref{dich} in the same manner we used \cref{rankreg} to show the
equidistribution of $(P_i(x))$ in \cref{atomsize}. Before we do so,
however, let us give a name to those $\Lambda$ for which the first
case of \cref{dich} holds.

\begin{definition}\label{dependency}
Given an affine constraint $A = (L_1, \dots, L_m)$ on $\ell$ variables
and integers $d, k > 0$ such that $d > k(p-1)$, the {\em
  $(d,k)$-dependency set of $A$} is the set of tuples $(\lambda_1,
\dots, \lambda_m) \in [0, p^{k+1}-1]$ such that $\sum_{i = 1}^m \lambda_i
P(L_i(x_1, \dots, x_\ell)) \equiv 0$ for every polynomial $P: \F^n \to
\T$ of degree $d$ and depth $k$.
\end{definition}

\cref{dich} says that if $\cB$ is a regular factor, $P_{A,
  \cB, \Lambda} \equiv 0$ exactly when the first condition holds. In other words:
\begin{corollary}\label{zeros}
Fix an integer $C > 0$, tuples $(d_1, \dots, d_C) \in \Z_{>
  0}^C$ and $(k_1, \dots, k_C) \in \Z_{\geq 0}^C$, and an affine
constraint $(L_1, \dots, L_m)$ on $\ell$ variables. For $i \in [C]$,
let $\Lambda_i$ be the $(d_i,k_i)$-dependency set of $A$.

Then, for any polynomial factor $\cB = (P_1, \dots, P_C)$, where each
$P_i$ has degree $d_i$ and depth $k_i$, and  $\cB$ has rank $>
r_{\ref{arank}}\left(\max_i d_i, \frac{1}{2}\right)$, it is the case that
a tuple $(\lambda_{i,j})_{i \in [C], j \in [m]}$ satisfies
$$\sum_{i=1}^C \sum_{j = 1}^m \lambda_{i,j} P_i(L_j(x_1, \dots, x_\ell)) \equiv
0$$ if and only if for every $i \in [C]$, $(\lambda_{i,1} \pmod {p^{k_i
  + 1}}, \dots, \lambda_{i, m}\pmod {p^{k_i+1}}) \in \Lambda_i$.
\end{corollary}
\begin{proof}
The ``if'' direction is obvious. For the ``only if'' direction, we use
Theorem \ref{dich} to conclude that if $\sum_{i, j}
\lambda_{i,j} P_i(L_j(\cdot)) \equiv 0$, it must be
that for every $i \in [C]$, $\sum_j \lambda_{i,j}
Q_i(L_j(\cdot)) \equiv 0$ for any polynomial $Q_i$ with degree
$d_i$ and depth $k_i$. This is equivalent to saying $(\lambda_{i,1}
\pmod {p^{k_i + 1}}, \dots, \lambda_{i,m}\pmod {p^{k_i + 1}}) \in
\Lambda_i$.
\end{proof}
\begin{remark}
For large characteristic fields, Hatami and Lovett \cite{HL11} showed that the
analog of \cref{zeros} is true even without the rank condition.
\end{remark}

The distribution of $(P_i(L_j(x_1,\dots,x_\ell)))$ is only going to be
supported on atoms which respect the constraints imposed by
dependency sets. This is obvious: if $P$ is a polynomial of degree $d$
and depth $k$, $(\lambda_1,\dots, \lambda_m)$ are in the
$(d,k)$-dependency set of $(L_1,\dots, L_m)$, and
$P(L_j(x_1,\dots,x_\ell)) = b_{j}$, then $\sum_j \lambda_j b_j =
0$. We call atoms which respect this constraint for all $P_i$ in a
factor {\em consistent}. Formally:
\begin{definition}[Consistency]\label{consistent}
Let $A$ be an affine constraint of size $m$.
A sequence of elements $b_1, \dots, b_m \in \T$ are said to be {\em
  $(d,k)$-consistent with $A$} if $b_1, \dots, b_m \in \U_{k+1}$ and
for every tuple $(\lambda_1, \dots, \lambda_m)$ in the
$(d,k)$-dependency set of $A$, it holds that $\sum_{i=1}^m \lambda_i
b_i = 0$.

Given vectors $\mv{d}=(d_1, \dots, d_C)\in \Z_{>0}^C$ and $\mv{k} = (k_1,\dots,k_C) \in \Z_{\geq 0}^C$, a
sequence of vectors $b_1, \dots, b_m \in \T^C$ are said to be {\em $(\mv{d},\mv{k})$-consistent with $A$} if for every $i \in [C]$, the
elements $b_{1,i},\dots, b_{m,i}$ are $(d_i,k_i)$-consistent with $A$.

If $\cB$  is a polynomial factor, the term {\em $\cB$-consistent with $A$} is a synonym for  {\em $(\mv{d},\mv{k})$-consistent with $A$} where  $\mv{d}=(d_1, \dots, d_C)$ and $\mv{k} = (k_1,\dots,k_C)$ are respectively the degree and depth sequences of polynomials defining $\cB$.  
\end{definition}

Now, the proof of equidistribution of $(P_i(L_j(x_1,\dots,x_\ell))$ is
straightforward.
\begin{theorem}\label{affequid}
Given $\eps > 0$, let $\cB$ be a polynomial factor of
degree $d > 0$,  complexity $C$, and rank
$r_{\ref{rankreg}}(d,\eps)$,  that is defined by a tuple of
polynomials $P_1, \dots, P_C: \F^n
\to \T$ having respective degrees $d_1, \dots, d_C$ and respective
depths $k_1, \dots, k_C$.
Let $A=(L_1, \dots, L_m)$ be an affine constraint on $\ell$ variables.

Suppose $b_1, \dots, b_m \in \T^{C}$  are atoms of $\cB$ that are
$\cB$-consistent with $A$. Then
$$
\Pr_{x_1, \dots, x_\ell}[\cB(L_j(x_1, \dots, x_\ell)) = b_{j}~
 \forall j \in [m]] =
\frac{\prod_{i=1}^C  |\Lambda_i|}{\|\cB\|^m} \pm \eps
$$
where $\Lambda_i$ is the $(d_i,k_i)$-dependency set of $A$.
\end{theorem}
\begin{proof}
The proof is similar to that of \cref{atomsize}.
\begin{align*}
&\Pr_{x_1, \dots, x_\ell}[P_i(L_j(x_1, \dots, x_\ell)) = b_{i,j}~
\forall i \in [C], \forall j \in [m]]\\
&= \E_{x_1, \dots, x_\ell} \left[ \prod_{i, j}\frac{1}{p^{k_i+1}} \sum_{
    \lambda_{i, j}=0}^{p^{k_i+1}-1} \expo{\lambda_{i, j}
    (P_{i}(L_j(x_1, \dots, x_\ell)) - b_{i,j})}\right]\\
&= \left(\prod_{i}p^{-(k_i+1)}\right)^m \sum_{(\lambda_{i,j})
\atop \in \prod_{i, j} [0,p^{k_i + 1}-1]}
\expo{-\sum_{i,j} \lambda_{i,j}b_{i,j}}
\E \left[\expo{\sum_{i,j} \lambda_{i,j}P_{i}(L_j(x_1,
    \dots, x_\ell)}\right]\\
&= p^{-m \sum_{i=1}^C(k_i+1)} \cdot \left(\prod_{i=1}^C |\Lambda_i|~ \pm~
  \eps p^{m \sum_{i=1}^C(k_i+1)}\right),
\end{align*}
where the last line follows because by \cref{zeros},
$\sum_{i,j}\lambda_{i,j} P_i(L_j(\cdot))$ is identically zero
for $\prod_i |\Lambda_i|$ many tuples $(\lambda_{i,j})$ and, in
that case, $\sum_{i,j}\lambda_{i,j} b_{i,j} = 0$ because
of the consistency requirement. For any other tuple
$(\lambda_{i,j})$, the expectation in the third line is bounded by
$\eps$ in absolute value.
\end{proof}

\section{Degree-structural Properties}\label{sec:degstruct}

In this section, we prove \cref{thm:degstruct} in
the introduction stating that if $\mcl{P}$ is degree-structural
(recall \cref{def:weakdegstruct}), then $\mcl{P}$ is locally characterized. The proof uses many
of the tools established in \cref{sec:equid}.\\

\noindent \textbf{Theorem \ref{thm:degstruct} (restated).}
{\em Every degree-structural property with bounded scope and max-degree is
a locally characterized affine-invariant property.}
\begin{proof}
Let $\mcl{P}$ be a degree-structural property with scope $\sigma$ and
max-degree $\Delta$. Denote
by $S$ the  set of tuples $(c,\mv{d},\Gamma)$ such that $c \leq \sigma$ and
$\mcl{P}$ is the union over all $(c,\mv{d},\Gamma) \in S$ of
$(c,\mv{d},\Gamma)$-structured functions.
It is clear that $\mcl{P}$ is affine-invariant, as having  degree bounded by a constant is
an affine-invariant property. It is also immediate that $\mcl{P}$ is
closed under taking restrictions to subspaces, since if $F$ is
$(c,\mv{d},\Gamma)$-structured, then $F$ restricted to any hyperplane
is also $(c,\mv{d},\Gamma)$-structured. The non-trivial part of the
theorem is to show that the locality is bounded. In other words we need to show that there is a constant $K$ such that for $n \geq
K$, if $F: \F^n \to \T$ is a function with  $F|_A \in \mcl{P}$ for every hyperplane $A \leq
\F^n$, then $F \in \mcl{P}$.

First, let us bound the degree of $F$. We know that  $F|_A \in \cP$ for every
hyperplane $A$. Therefore, $\deg(F|_A) \leq p\sigma\Delta$ for every
$A$, as $F|_A$ is a function of at most $\sigma$ polynomials each of degree at most $\Delta$ over a
field of characteristic $p$. It follows that $F$ itself is of degree
$\leq  p\sigma\Delta$.

Let $r: \Z_{>0} \to \Z_{>0}$ be a function to be fixed later.  Define $r_2:
\Z_{>0} \to \Z_{>0}$ so that $r_2(m) > r(C_{\ref{factorreg}}^{(r,
  p\sigma\Delta)}(m+\sigma)) + C_{\ref{factorreg}}^{(r,
  p\sigma\Delta)}(m+\sigma) + p$.

We apply \cref{factorreg}to $\set{F}$ to find an $r_2$-regular
polynomial factor $\cB$ of degree $\leq p\sigma\Delta$, defined by
polynomials $R_1, \dots, R_C: \F^n \to \T$, where $C \leq
C_{\ref{factorreg}}^{r_2,d}(1)$. Since $F$ is measurable with
respect to $\cB$,  there exists a function $\Sigma: \T^C\to \F$,
such that $F(x) = \Sigma(R_1(x),\dots,R_C(x))$.

From each $R_i$ pick a monomial with degree equal to $\deg(R_i)$ and a monomial (possibly the same one) with depth equal to $\mrm{depth}(R_i)$.
By taking $K$ to be sufficiently large, we can gaurantee the existence of an $i_0 \in [n]$ such that $x_{i_0}$ is  not involved in any of these monomials. Consequently $\deg(R_i') = \deg(R_i)$ and $\mrm{depth}(R_i') = \mrm{depth}(R_i)$ for all $i \in [C]$, where $R_1', \dots, R_C'$ are  the restrictions of $R_1,\dots,R_C$, respectively,  to the hyperplane $\{x_{i_0} = 0\}$. Also by \cref{rankrestrict}, $R_1',\dots,R_C'$ have rank $> r_2(C)-p$.
Since $F|_{x_{i_0} = 0} \in \mcl{P}$, by definition of $\mcl{P}$, there must exist $(c,\mv{d},\Gamma) \in S$ with $c \leq \sigma$
such that
$$\Sigma(R_1',\dots,R_C') = \Gamma(P_1,\dots,P_{c}),$$
where $\deg(P_i) \leq d_i$ for all $i \in [c]$.

Now, apply \cref{factorreg} to find an $r$-regular refinement of
the factor defined by the tuple of polynomials $(R_1', \dots,
  R_C', P_1, \dots, P_c)$. Because of our choice of $r_2$ and the last
  part of \cref{factorreg}, we obtain a syntactic refinement of
  $\set{R_1', \dots,  R_C'}$. That is, we obtain
a tuple $\cB'$ of polynomials $R_1',\dots,R_C', S_1,\dots,
  S_D: \F^n \to \T$ such that it has degree $\leq p\sigma\Delta$, its rank $> r(C+D)$, and $C+D\leq
C_{\ref{factorreg}}^{(r)}(C+\sigma)$, and for each $i \in [c]$, $P_i =
\Gamma_i(R_1',\dots,R_C',S_1,\dots,S_D)$ for some function $\Gamma_i :
\T^{C+D} \to \T$.  So for all $x \in \F^n$,
\begin{align*}
&\Sigma(R_1'(x),\dots,R_C'(x)) =\\
&\quad \quad \Gamma(\Gamma_1(R_1'(x),\dots,R_C'(x),S_1(x),\dots,S_D(x)), \dots,
\Gamma_c(R_1'(x),\dots,R_C'(x),S_1(x),\dots,S_D(x))).
\end{align*}
Applying \cref{atomsize}, we see that if the rank of $\cB'$ is $>
r_{\ref{atomsize}}\left(p\sigma\Delta, \varepsilon\right)$ where $\varepsilon>0$ is sufficiently small (say $\varepsilon=\|\cB'\|/2$), then
$(R_1'(x), \dots, R_C'(x),$ $S_1(x), \dots, S_D(x))$ acquires every value
in its range. Thus, we have the identity
$$\Sigma(a_1, \dots, a_c) = \Gamma(\Gamma_1(a_1, \dots, a_C, b_1,
\dots, b_D), \dots, \Gamma_c(a_1, \dots, a_C, b_1, \dots, b_D)),$$
for every $a_i \in \U_{\mrm{depth}(R_i') + 1}$ and $b_i \in
\U_{\mrm{depth}(S_i) + 1}$. Thus, we can substitute $R_i$ for $R_i'$
and   $0$ for $S_i$ in the above equation and still retain
  the identity
\begin{align*}
F(x) &= \Sigma(R_1(x),\dots,R_C(x))\\
&=\Gamma(\Gamma_1(R_1(x),\dots,R_C(x),0,\dots,0), \dots,\Gamma_c(R_1(x),\dots,R_C(x),0,\dots,0))\\
&= \Gamma(Q_1(x), \dots, Q_c(x))
\end{align*}
where $Q_i : \F^n \to \T$ are defined as $Q_i(x) =
\Gamma_i(R_1(x),\dots,R_C(x),0,\dots,0)$.
Since for every $i$, $\deg(R_i) = \deg(R_i')$ and $\mrm{depth}(R_i)
= \mrm{depth}(R_i')$, we can apply \cref{comp} below to conclude that
$\deg(Q_i) \leq \deg(P_i) \leq d_i$ for every $i \in [c]$, as long as
the rank of $\cB'$ is $ > r_{\ref{comp}}(p\sigma\Delta)$.  Finally, we
show that $Q_1, \dots, Q_c$ map to $\U_1 = \iota(F)$ and, so, are
classical. Indeed, since $P_1, \dots, P_c$ are classical,
$\Gamma_1, \dots, \Gamma_c$ must map to $\iota(F)$ on all of
$\prod_{i=1}^C \U_{\mrm{depth}(R_i') + 1} \times \prod_{i=1}^D
\U_{\mrm{depth}(S_i) + 1} \supseteq \prod_{i=1}^C \U_{\mrm{depth}(R_i)
  + 1} \times \set{0}^D$.
Hence, $F \in \cP$.
\end{proof}

The following theorem, used in the proof above, shows that a function
of a high rank collection of polynomials has the degree one would
expect. Thus, it displays yet another way in which high-rank
polynomials behave ``generically''. The proof is via another
application of the near-orthogonality result in \cref{dich}.
\begin{theorem}\label{comp}
For an integer $d > 0$, let $P_1, \dots, P_C: \F^n \to \T$ be
polynomials of degree $\leq d$ and rank $> r_{\ref{comp}}(d)$,
and let $\Gamma: \T^C \to \T$ be an arbitrary function. Define the
polynomial $F:
\F^n \to \T$ by $F(x) = \Gamma(P_1(x), \dots, P_C(x))$. Then, for every
collection of polynomials $Q_1, \dots Q_C: \F^n \to \T$ with
$\deg(Q_i) \leq \deg(P_i)$ and $\mrm{depth}(Q_i) \leq
\mrm{depth}(P_i)$ for all $i \in [C]$, if $G: \F^n \to \T$ is the
polynomial $G(x) = \Gamma(Q_1(x), \dots, Q_C(x))$, it holds that
$\deg(G) \leq \deg(F)$.
\end{theorem}
\begin{proof}
Let $f(x) = \expo{F(x)}$ and $\gamma(x_1, \dots, x_C) =
\expo{\Gamma(x_1, \dots, x_C)}$. Let $D = \deg(F)$. Then, for every
$x, y_1, \dots, y_{D+1} \in \F^n$,
$$\Delta_{y_{D+1}}\cdots \Delta_{y_1} f(x) = 1.$$
We need to show that $g(x) = \expo{G(x)}$ also satisfies $\Delta_{y_{D+1}}\cdots
\Delta_{y_1} g(x) = 1$.

Let $k_1, \dots, k_C$ be the depths of $P_1, \dots, P_C$,
respectively. Then, each $P_i$ takes values in $\U_{k_i+1}$. Let $\Sigma$ denote the group $\Z_{p^{k_1 + 1}}
\times \cdots \times \Z_{p^{k_C + 1}}$. Considering the Fourier transform of $\gamma$, we have
$$
f(x) = \gamma(P_1(x), \dots, P_C(x))= \sum_{\beta \in \Sigma} \hat{\gamma}(\beta) \expo{\sum_{i=1}^C
  \beta_i P_i(x)}.
$$
Next, we look at the derivative.
\begin{align*}
\Delta_{y_{D+1}}\cdots \Delta_{y_1} f(x)
&= \Delta_{y_{D+1}}\cdots \Delta_{y_1}
\left( \sum_{\beta \in \Sigma} \hat{\gamma}(\beta) \expo{\sum_{i=1}^C
  \beta_i P_i(x)}\right)\\
&= \sum_{\alpha_J \in \Sigma: J \subseteq [D+1]} \prod_{J \subseteq
  [D+1]} \hat{\gamma}(\alpha_J) \expo{(-1)^{|J|+1}\sum_{i=1}^C \alpha_{J,i}
  P_i \left(x + \sum_{j \in J} y_j\right)}.
\end{align*}
Denoting $\delta(\alpha)=\prod_{J \subseteq [D+1]} \hat{\gamma}(\alpha_J)$ for
$\alpha=(\alpha_J)_{ J \subseteq [D+1]} \in \Sigma^{\mathcal{P}([D+1])}$, we have
\begin{align}\label{eqn:deriv}
\Delta_{y_{D+1}}\cdots \Delta_{y_1} f(x)
= \sum_{\alpha \in \Sigma^{\mathcal{P}([D+1])}} \delta(\alpha)
\expo{\sum_{i=1}^C \sum_{J \subseteq [D+1]} (-1)^{|J|+1} \alpha_{J,i}
  P_i\left(x + \sum_{j \in J} y_j\right)}.
\end{align}
For any $i$, if there is a $J$ such that $|J| > \deg(\alpha_{J,i} P_i)$,
we can use \cref{eqn:poly} to rewrite $\alpha_{J,i}
  P_i\left(x + \sum_{j \in J} y_j\right)$ as a linear combination
  (over $\Z$) of $\set{ P_i\left(x+\sum_{j \in J'} y_j\right) :
    |J'| < |J|}$. We repeat this process  until for
  every $i$ and $J$, either $\alpha_{J,i} = 0$ or $|J| \leq
  \deg(\alpha_{J,i} P_i)$.  Denoting by $\cA$ the set of $\alpha \in \Sigma^{\cP([D+1])}$ that satisfy this condition, we have obtained a new set of   coefficients  $\delta'(\alpha)$  such that
\begin{align*}
\Delta_{y_{D+1}}\cdots \Delta_{y_1} f(x)
= \sum_{\alpha \in \cA} \delta'(\alpha)
\expo{\sum_{i=1}^C \sum_{J \subseteq [D+1]}\alpha_{J,i}
  P_i\left(x + \sum_{j \in J} y_j\right)}.
\end{align*}
Now, the crucial observation is that if instead of $P_1, \dots,
P_C$, we had $Q_1, \dots, Q_C$, the same decomposition applies.
\begin{align}\label{eqn:forg}
\Delta_{y_{D+1}}\cdots \Delta_{y_1} g(x)
= \sum_{\alpha \in \cA}  \delta'(A)
\expo{\sum_{i=1}^C \sum_{J \subseteq [D+1]}\alpha_{J,i}
  Q_i\left(x + \sum_{j \in J} y_j\right)}.
\end{align}
The reason is that \cref{eqn:deriv} remains valid as is if $f$ is
replaced by $g$ and the $P_i$'s are replaced by $Q_i$'s, and furthermore
since $\deg(P_i) \geq \deg(Q_i)$ and $\mrm{depth}(P_i)
\geq \mrm{depth}(Q_i)$, the applications of \cref{eqn:poly} remain
valid also. Therefore, \cref{eqn:forg} is also valid \footnote{Note
  that in \cref{eqn:forg}, one could have nonzero $\alpha_{J,i}$ and
  $|J| > \deg(\alpha_{J,i} Q_i)$, for  $A=(\alpha_J)_{J \subseteq [D+1]}$ with $\delta'(A)\neq 0$.}.

But now, we argue that $\delta'(\cdot)$ are uniquely determined.  Let $k =
\max_i k_i \leq d/(p-1)$.
\begin{claim}
If $P_1,\dots, P_C$ are of rank $> r_{\ref{arank}}(d,
\nfrac{1}{|\cA|})+1$, the functions
$$
\set{\expo{\sum_{i=1}^C
    \sum_{J \subseteq [D+1]}\alpha_{J,i}   P_i\left(x + \sum_{j \in J}
      y_j\right)} \ :\  \alpha \in \cA}
$$
are linearly independent over $\C$.
\end{claim}
\begin{proof}
Note that all these functions have $L^2$-norm equal to $1$. Hence  it suffices to show that their pairwise inner products are all bounded in absolute value by $1/ |\cA|$. To prove this consider $\alpha, \beta \in \cA$, and note that by \cref{dich} and, in particular, \cref{rmk:small}, unless all the $\alpha_{J,i}-\beta_{J,i}$ are zero, $$\left|\E\left[\expo{\sum_{i=1}^C \sum_{J \subseteq [D+1]}(\alpha_{J,i}-\beta_{J,i})  P_i\left(x + \sum_{j \in J} y_j\right)}\right] \right| < \frac{1}{|\cA|}.$$
\end{proof}
Therefore, since $\Delta_{y_{D+1}}\cdots \Delta_{y_1} f(x) = 1$, we
must have $\delta'(\alpha) = 1$ when $\alpha$ is the all-zero tuple, and $\delta'(\alpha) = 0$ for
every nonzero $\alpha$. Plugging into \cref{eqn:forg}, we get
$\Delta_{y_{D+1}}\cdots \Delta_{y_1} g(x) = 1$.
\end{proof}

\section{Property Testing}\label{sec:testing}

\subsection{Decomposition Theorems}

Decomposition theorems are a major class of theorems in additive and extremal combinatorics. These are statement that tell us that a function $f$ with certain properties can be decomposed as a sum $\sum_{i=1}^k g_i$, where the functions $g_i$ have certain other properties.
We have already seen a decomposition theorem in \cref{arank}: if a
polynomial $P: \F^n \to \T$ of degree $\leq d$ satisfies
$\|\expo{P}\|_{U^d} \geq \eps$, then there exists a factor $\cB$ of
complexity $\leq r_{\ref{arank}}(d,\eps)$ such that $P$ is a function
of the polynomials defining $\cB$. 

In this section, we discuss decomposition theorems of a particular type called approximate structure theorems. These are  results that say that, under appropriate conditions, we can write a function $f$ as $f_1 + f_2$, where $f_1$ is ``structured'' in some sense, and $f_2$ is ``quasirandom''. 
The rough idea is that the structure of $f_1$ is strong enough for us to be 
able to analyze it reasonably explicitly, and the quasirandomness of $f_2$ is strong enough
for many properties of $f_1$ to be unaffected if we ``perturb'' it to $f=f_1 + f_2$. Often, in order to
obtain stronger statements about the structure and the quasirandomness, one allows also a small $L^2$-error: that is, one writes $f$ as $f_1 + f_2 + f_3$ with $f_1$ structured, $f_2$ quasirandom, and $f_3$ small in $L^2$ .

\ignore{
To describe such decompositions, we need to consider conditional expectations over polynomial factors.
Note that a polynomial factor defines a partition of $\F^n$, and thus one can consider the conditional expectation of functions $f:
F^n \to \C$ with respect to this partition:

\begin{definition}[Expectation over polynomial factor]\label{condexp}
Given a factor $\cB$ and a function $f:\F^n\to\bits$, the
{\em expectation} of $f$ over an atom $y \in \T^{|\cB|}$, denoted by $\E[f|y]$, is the average
$f(x)$ over $\{x:\cB(x)=y\}$. The {\em
  conditional expectation} of $f$ over $\cB$, is the real-valued
function over $\F^n$ given by $\E[f|\cB](x)=\E[f|\cB(x)]$. In
particular, it is constant on every atom of the polynomial factor and
hence it is a function of the polynomials defining $\cB$.
\end{definition}}

The Strong Decomposition Theorem below shows that any Boolean function can
be decomposed into the sum of a conditional expectation over a high
rank factor, a function with small Gowers norm, and a function with
small $L^2$-norm.

\begin{theorem}[Strong Decomposition Theorem; Theorem 4.4 of \cite{BFL12}]\label{thm:strongdecomp}
Suppose $\delta > 0$ and $ d \geq 1$ are integers. Let $\eta: \N
\to \R^+$ be an arbitrary non-increasing function and $r: \N \to \N$ be an arbitrary
non-decreasing function. Then there exist $N =
N_{\ref{thm:strongdecomp}}(\delta, \eta, r, d)$ and $C =
C_{\ref{thm:strongdecomp}}(\delta,\eta,r,d)$ such that the following holds.

Given $f: \F^n \to \bits$ where $n > N$, there
exist three functions $f_1, f_2, f_3: \F^n \to
\R$ and a polynomial factor  $\cB$ of
degree at most $d$ and complexity at most $C$ such that the following conditions hold:
\begin{itemize*}
\item[(i)]
$f=f_1+f_2+f_3$.
\item[(ii)]
$f_1 = \E[f|\cB]$.
\item[(iii)]
$\|f_2\|_{U^{d+1}} \leq 1/\eta(|\cB|)$.
\item[(iv)]
$\|f_3\|_2 \leq \delta$.
\item[(v)]
$f_1$ and $f_1 + f_3$ have range $[0,1]$; $f_2$ and $f_3$ have range $[-1,1]$.
\item[(vi)]
$\cB$ is $r$-regular.
\end{itemize*}
\ignore{
Moreover, if $\cB_0$ is a syntactic refinement
of some $\hat{\cB}$ of rank at least $r(C)+C$, then $\cB$ will also
be a syntactic refinement of $\hat{\cB}$ (in particular,  if $\cB_0=\hat{\cB}$).}
\end{theorem}

It turns out though that this Strong Decomposition Theorem is
not quite sufficient for the purpose of this paper. The issue is that the bound on $f_3$ above
is a constant $\delta$. Ideally, we would want $\delta$ to
decrease as a function of the complexity of the polynomial
factor, but such a decomposition theorem is simply not true. However, analogous to what it is shown in \cite{AFKS} in the context graphs, here one can find two polynomial factors $\cB' \succeq_{syn} \cB$ such that, the structured part $f_1$ equals to  $\E[f|\cB']$, but now the $L^2$-norm  of $f_3$ can be made arbitrarily small in terms of the complexity of the coarser factor $\cB$. Furthermore for most atoms $c$ of $\cB$, the function $f:\F^n \to \bits$ have roughly the same density on $c$ and most of its subatoms in $\cB'$. To make this precise, we make the following definition.

\begin{definition}[Polynomial factor represents another factor]\label{def:rep}
Given a function $f: \F^n \to \bits$, a polynomial factor $\cB'$
that refines another factor $\cB$ and a real $\zeta \in (0,1)$, we
say {\em $\cB'$ $\zeta$-represents $\cB$ with respect to $f$} if
for at most $\zeta$ fraction of atoms $c$ of $\cB$, more than
$\zeta$ fraction of the atoms $c'$ lying inside $c$ satisfy
$|\E[f|c]-\E[f|c']|>\zeta$.
\end{definition}

We  can now state the following  Super Decomposition Theorem proven in~\cite{BFL12}.

\begin{theorem}[Super Decomposition Theorem; Theorem 4.9 of \cite{BFL12}]\label{thm:superdecomp}
Suppose $\zeta > 0$ is a real and $ d, C_0 \geq 1$ are integers. Let $\eta: \N \to \R^+$ and
$\delta:\N\to\R^+$ be arbitrary non-increasing functions, and
$r:\N\to\N$ be an arbitrary non-decreasing function. Then there exist
$N =  N_{\ref{thm:superdecomp}}(\delta,\eta,r,d,\zeta)$ and $C =
C_{\ref{thm:superdecomp}}(\delta,\eta,r,d,\zeta)$ such that
the following holds.

Given $f: \F^n \to \bits$ where $n > N$, there
exist functions $f_1, f_2, f_3: \F^n \to \R$, and polynomial factors
$\cB' \succeq_{syn} \cB$ of degree  at most $d$ and of complexity at most
$C$, such that the following conditions hold:
\begin{itemize*}
\item[(i)]
$f=f_1+f_2+f_3$.
\item[(ii)]
$f_1 = \E[f|\cB']$.
\item[(iii)]
$\|f_2\|_{U^{d+1}} \leq \eta(|\cB'|)$.
\item[(iv)]
$\|f_3\|_2 \leq \delta(|\cB|)$.
\item[(v)]
$f_1$ and $f_1 + f_3$ have range $[0,1]$; $f_2$ and $f_3$ have
range $[-1,1]$.
\item[(vi)]
$\cB$ and $\cB'$ are both $r$-regular.
\item[(vii)]
$\cB'$ $\zeta$-represents $\cB$ with respect to $f$.
\end{itemize*}
\end{theorem}

Although the above Super Decomposition Theorem may be useful by itself
for other applications, we will need a particular variant. The factor $\cB'$ is a syntactic refinement of $\cB$, and thus is defined by adding new polynomials $Q_1,\ldots,Q_{|\cB'|-|\cB|}$ to the polynomials defining $\cB$. Then for each atom $c$ in the coarser factor $\cB$ we will select one atom $c'$ of $\cB'$ such that the following hold:
\begin{itemize}
\item There is a \emph{fixed}  $s \in \T^{|\cB'|-|\cB|}$ such that for every atom $c$ in $\cB$ its corresponding atom  $c'$ is obtained by requiring  $(Q_1,\ldots,Q_{|\cB'|-|\cB|})$ to be equal to  $s$. 
\item The $L^2$-norm of $f_3$ conditioned inside every such  atom (i.e., $\E_{x \in
c'}[|f_3(x)|^2]$) is small. 
\item Most subatoms $c'$ will ``well-represent'' (in the sense of \cref{def:rep}) their corresponding atoms $c$
  from $\cB$.
\end{itemize}

Before stating this formally, let us also take this opportunity to
remark that it is possible to adapt the proofs of the above decomposition theorems to decompose  several
functions $f^{(1)},\ldots,f^{(R)}:\F^n\to\{0,1\}$ simultaneously. Alternatively, this could
be thought of as decomposing a single vector-valued function
$f:\F^n\to\{0,1\}^R$.  Now we finally state the decomposition theorem that we will use in the proof of our main result. 
 
\begin{theorem}[Subatom Selection; Theorem 4.12 of \cite{BFL12}]\label{thm:subatom2}
Suppose $\zeta > 0$ is a real and $d,R \geq 1$ are integers. Let $\eta, \delta: \N \to \R^+$ be arbitrary
non-increasing functions, and let $r: \N \to \N$ be an arbitrary
non-decreasing function. Then, there exist $C =
C_{\ref{thm:subatom2}}(\delta, \eta,   r, \zeta, R)$ such that the
following holds.

Given $f^{(1)},\dots,f^{(R)}: \F^n \to \bits$, there exist functions $f^{(i)}_1, f^{(i)}_2, f^{(i)}_3 :
\F^n \to \R$ for all $i \in [R]$, a polynomial factor $\cB$ of degree
$d$ with atoms denoted by elements of $\T^{|\cB|}$, a syntactic
refinement $\cB' \succeq_{syn}\cB$ of degree $d$ with complexity at
most $C$ and atoms denoted by elements of $\T^{|\cB|} \times
\T^{|\cB'| - |\cB|}$, and an element $s \in \T^{|\cB'| - |\cB|}$ such
that the following is true:
\begin{itemize*}
\item[(i)]
$f^{(i)} = f_1^{(i)} + f_2^{(i)} + f_3^{(i)}$ for every $i \in [R]$.
\item[(ii)]
$f_1^{(i)} = \E[f^{(i)}| \cB']$ for every $i \in [R]$.
\item[(iii)]
$\|f_2^{(i)}\|_{U^{d+1}} < \eta(|\cB'|)$ for every $i \in [R]$.
\item[(iv)]
For every $i \in [R]$, $f_1^{(i)}$ and $f_1^{(i)} + f_3^{(i)}$ have range
$[0,1]$, and  $f_2^{(i)}$ and $f_3^{(i)}$ have range $[-1,1]$.
\item[(v)]
$\cB$ and $\cB'$ are both $r$-regular.
\item[(vi)]
For every atom $c \in \T^{|\cB|}$ of $\cB$, the subatom $c' = (c, s) \in
\T^{|\cB'|}$ satisfy
$$\E\left[\left. |f_3^{(i)}|^2 ~\right|~ (c,s) \right] < \delta(|\cB|)^2$$
for every $i \in [R]$.
\item[(vii)] If $c$ is an atom of $\cB$ chosen uniformly at random, then 
$$\Pr_{c}\left[\max_{i\in [R]} \left(\left|\E\left[\left. f^{(i)} \right|c\right] -
\E\left[\left. f^{(i)} \right|(c,s) \right]\right|\right)>\zeta\right]<\zeta.$$
\end{itemize*}
\end{theorem}

\subsection{Big Picture Functions}

Suppose we have a function $f: \F^n \to [R]$, and we want to find out
whether it induces a particular affine constraint $(A,\sigma)$, where
$A =(L_1,\dots,L_m)$ is a sequence of affine forms on $\ell$ variables
and $\sigma \in [R]^m$. Now, suppose $\F^n$ is partitioned by a
polynomial factor $\cB$ defined by polynomials $P_1,\dots, P_C$ of
degrees $d_1, \dots, d_C$ and depths $k_1, \dots, k_C$. Then, observe
that  if $b_1, \dots, b_m \in \T^C$ denote the atoms
of $\cB$ containing $L_1(x_1,\dots,x_\ell), \dots,
L_m(x_1,\dots,x_\ell)$ respectively, it must be the case that
$b_1, \dots, b_m$ are $\cB$-consistent with $A$  (as defined in
\cref{consistent}). Thus, to locate where $f$ might induce
$(A,\sigma)$, we should restrict our search to sequences of
atoms consistent with $A$.

It will be convenient to ``blur'' the given function $f$ so as to retain
only atom-level information about it. That is, for every atom $c$ of $\cB$, we will define $f_{\cB}(c)  \subseteq [R]$ to be the set of all values that $f$ takes within $c$. 
\begin{definition}
Given a function $f:\F^n\to [R]$ and a polynomial factor $\cB$,
the {\em big picture function} of $f$ is the function
$f_{\cB}:\T^{|\cB|}\to\cP([R])$, defined by $f_{\cB}(c)=\{f(x):\cB(x)=c\}$. 
\end{definition}

On the other hand, given any function $g:\T^C\to \cP([R])$, and 
a vector of degrees $\mv{d}=(d_1,\ldots,d_C)$ and depths
$\mv{k} = (k_1,\dots, k_C)$ (which we think of as
corresponding to the degrees and depths of some  polynomial
factor of complexity $C$), we will define what it means for such a
function to ``induce'' a copy of a given constraint.

\begin{definition}[Partially induce]\label{def:partial}
Suppose we are given vectors $\mv{d}=(d_1, \dots, d_C)\in \Z_{>0}^C$
and $\mv{k} = (k_1,\dots,k_C) \in \Z_{\geq 0}^C$, a
function
$g: \prod_{i\in [C]} \U_{k_i+1} \to \cP([R])$, and an induced affine constraint
$(A,\sigma)$ of size $m$. We say that {\em $g$ partially
  $(\mv{d},\mv{k})$-induces $(A,\sigma)$} if there exist a sequence $b_1,\ldots,b_m \in \T^C$ that is  $(\mv{d},\mv{k})$-consistent with $A$, and 
$\sigma_j \in g(b_j)$ for each $j \in [m]$.
\end{definition}

\cref{def:partial} is justified by the following trivial observation. 

\begin{remark}\label{rem:pinduce}
If $f:\F^n\to [R]$ induces a constraint $(A,\sigma)$, then for a
factor $\cB$ defined by polynomials of respective degrees
$(d_1,\ldots,d_{|\cB|}) = \mv{d}$  and respective depths $(k_1, \dots,
k_{|\cB|}) = \mv{k}$, the big picture function $f_{\cB}$
partially $(\mv{d},\mv{k})$-induces $(A,\sigma)$.
\end{remark}
\ignore{
\begin{proof}
Let $m$ be the size of $A$ and $\ell$ be its number of variables.
Suppose that $f$ induces $(A,\sigma)$ at $x_1,\dots,x_{\ell}$, and let
$c_1,\dots, c_{m} \in \F^{|\cB|}$ be the images of the $m$ atoms in
$\cB$ defined by $c_1 =
\cB(L_{1}(x_1,\dots,x_{\ell})), c_2 =
\cB(L_{2}(x_1,\dots,x_{\ell})),$ $\dots,$ $c_{m} =
\cB(L_{m}(x_1,\dots,x_{\ell}))$ where $A = (L_{1},\dots,
L_{m})$. Then, by definition of consistency, it must be the case that
$c_1,\dots,c_m$ are consistent with respect to $A$ and $d_1,\dots,d_{|\cB|}$. This fulfills the first condition of Definition \ref{def:partial},
and the second condition is true by the definition of every $f_{\cB}(c_i)$ including all
values that $f$ takes in that atom.
\end{proof}}

To handle a possibly infinite collection $\cA$ of affine constraints, we will employ a
compactness argument, analogous to one used in \cite{AS08a} to bound
the size of the constraint partially induced by the big picture function. Let us  make the
following definition:

\begin{definition}[The compactness function]\label{def:psi}
Suppose we are given positive integers $C$ and $d$, and a possibly infinite
collection of induced affine constraints $\mathcal{A} = \{(A^1,
\sigma^1), (A^2, \sigma^2), \dots\}$, where  $(A^i,\sigma^i)$ is of size $m_i$. For fixed $\mv{d} = (d_1, \dots, d_C) \in [d]^C$ and $\mv{k} = (k_1, \dots, k_C) \in \left[0,\left \lfloor \frac{d-1}{p-1}\right \rfloor
\right]^C$, denote by $\cG(\mv{d},\mv{k})$ the set of functions
$g: \prod_{i=1}^C \U_{k_i+1} \to \cP([R])$ that partially $(\mv{d},\mv{k})$-induce some $(A^i,\sigma^i)\in
\cA$. The compactness function is defined as  
\begin{align*}
\Psi_\cA(C,d) = \max_{\mv{d},\mv{k}} \max_{g \in
  \cG(\mv{d},\mv{k})} \min_{(A^i,\sigma^i) \text{ {\em partially}} \atop
  \text{{\em $(\mv{d},\mv{k})$-induced by }} g} m_i
\end{align*}
where the outer $\max$ is over vectors $\mv{d} = (d_1, \dots, d_C) \in [d]^C$ and $\mv{k}
= (k_1, \dots, k_C) \in \left[0,\left \lfloor \frac{d-1}{p-1}\right \rfloor
\right]^C$. Whenever $\cG(\mv{d},\mv{k})$ is empty, we set the corresponding maximum to $0$.
\end{definition}

Note that $\Psi_\cA(C,d)$ is indeed finite, as  the number of possible
degree and depth sequences are bounded by $d^{2C}$, and the size of
$\cG(d_1,\dots,d_C)$ is bounded by $2^{Rp^{dC}}$.

\begin{remark}\label{rem:induce}
Note that if a function $g: \T^C \to \cP([R])$ partially
$(\mv{d},\mv{k})$-induces some constraint from $\cA$ where $\mv{d} \in [d]^C$, then  $g$ must belong to $\cG(\mv{d},\mv{k})$, and consequently it will necessarily partially  induce some $(A^i,\sigma^i) \in \cA$ whose size is at most $\Psi_\cA(C,d)$.
\end{remark}

\subsection{Proof of Testability}

We prove the main result, \cref{thm:main3}, in this section. In
fact, we will show the following.
\begin{theorem}\label{mainthm2}
Let $d>0$ be an integer. 
Suppose we are given a possibly infinite collection of  affine
constraints $\mathcal{A} = \{(A^1, \sigma^1), (A^2, \sigma^2),\dots\}$ where each
$(A^i,\sigma^i)$ is an affine constraint of complexity $\leq d$, and of size $m_i$ on $\ell_i$
variables. Then, there are functions $\ell_\cA:(0,1) \to \Z_{>0}$ and
$\delta_\cA:(0,1) \to (0,1)$ such that the following is true for any $\eps
\in (0,1)$.  If a function $f: \F^n \to [R]$  is $\eps$-far from being $\cA$-free, then 
$f$ induces at least  $\delta_\cA(\eps) p^{n\ell_i}$ many copies of some $(A^i,\sigma^i)$ with $\ell_i < \ell_\cA(\eps)$ .

Moreover, if $\cA$ is locally characterized, then $\ell_A(\eps)$ is a
constant independent of $\eps$.
\end{theorem}

\cref{thm:main3} immediately follows. Consider the following test: choose
uniformly at random $x_1,\dots,$ $x_{\ell_\cA(\eps)}$ $\in$ $\F^n$, let $H$
denote the affine space $\left\{x_1+\sum_{j=2}^{\ell_\cA(\eps)} \lambda_j x_j : \lambda_j
\in \F \right\}$, and check whether $f$ restricted to $H$ is $\cA$-free
or not, thus making $\leq p^{\ell_\cA(\eps)}$ queries. By
\cref{mainthm2}, if $f$ is $\eps$-far from $\cA$-freeness, this test
rejects with probability at least $\delta_\cA(\eps)$.
~\\

\begin{proofof}{\cref{mainthm2}}

\paragraph{Preliminaries.}
Fix a function $f: \F^n \to [R]$ that is $\eps$-far from being
$\cA$-free. For $i \in [R]$, define $f^{(i)}:\F^n \to \bits$ so that
$f^{(i)}(x)$ equals $1$ when $f(x) = i$ and equals $0$
otherwise. Additionally, set the following parameters, where
$\Psi_\cA$ is the compactness function from \cref{def:psi}:
{\allowdisplaybreaks
$$
\begin{array}{lclclclclcl}
\alpha(C) &=& p^{-2dC\Psi_\cA(C,d)},&\qquad& \rho(C) &=& r_{\ref{arank}}(d,\alpha(C)), &\qquad&\zeta &=& \frac{\eps}{8R},\\
\Delta(C) &=& \frac{1}{16}\zeta^{\Psi_\cA(C,d)},& &
\eta(C) &=& \frac{1}{8p^{dC\Psi_\cA(C,d)}}
\left(\frac{\eps}{24R}\right)^{\Psi_\cA(C,d)}.&&&&\\
\end{array}
$$
}

\paragraph{Decomposing by regular factors.}
Next, apply Theorem \ref{thm:subatom2} to the functions
$f^{(1)},f^{(2)}, \dots, f^{(R)}$ in order to get polynomial factors
$\cB' \succeq_{syn} \cB$ of complexity at most
$C_{\ref{thm:subatom2}}(\Delta, d, \rho, \zeta, \eta)$, an element $s \in
\T^{|\cB'|-|\cB|}$, and functions
$f_1^{(i)}, f_2^{(i)}, f_3^{(i)}: \F^n \to \R$ for each $i \in
[R]$ with the desired properties. The sequence of polynomials generating $\cB'$ will be denoted
by $P_1,\dots,P_{|\cB'|}$. Since $\cB'$ is a syntactic refinement, we
can assume  $\cB$ is generated by the polynomials $P_1,\dots,P_{|\cB|}$.
Let $C = |\cB|$ and $C' = |\cB'|$. Note that 
$\|\cB\| < p^{(k_{\max}+1)C} \leq p^{dC}$, where $k_{\max} \leq \left
  \lfloor (d-1)/(p-1)\right\rfloor$ is the maximum depth
of a polynomial in $\cB$.  Denote the degree of $P_i$ by
$d_i$ and the depth of $P_i$ by $k_i$.

\paragraph{Cleanup.}
Based on $\cB'$ and $\cB$, we  construct a
function $F: \F^n \to [R]$ that is $\frac{\eps}{2}$-close to $f$ and hence,
still violates $\cA$-freeness.  The ``cleaner'' structure of $F$ will help us locate the
induced constraint violated by $f$. 

The function $F$ is the same as $f$ except for the following: For every  atom $c$ of $\cB$, let $t_c=\arg\max_{j \in
[R]}\Pr[f(x)=j~|~\cB'(x)=(c,s)]$ be the most popular value inside the corresponding subatom $(c,s)$.
 
\begin{itemize}
\item {\bf Poorly-represented atoms:} If there exists $i \in [R]$ such that $|\Pr[f(x)=i~|~\cB(x)=c]-\Pr[f(x)=i~ |~ \cB'(x) = (c,s)]|>\zeta$, then set 
$F(z) = t_c$ for every $z$ in the atom $c$.

\item {\bf Unpopular values:} Otherwise, for any $z$ in the atom $c$ with $0< \Pr_x[f(x)=f(z)~|~\cB'(x) = (c,s)]<\zeta$, set $F(z)=t_c$.
\end{itemize}

A key property of the cleanup function $F$ is that it 
supports a value inside an atom $c$ of $\cB$ only if the original
function $f$ acquires the value on at least an $\zeta$ fraction of
the subatom $(c,s)$. Furthermore as the following lemma shows it is $\eps/2$-close to $f$, and therefore, it is not $\cA$-free.

\begin{lemma}
The cleanup function $F$ is $\eps/2$-close to $f$, and therefore, it is not $\cA$-free.
\end{lemma}
\begin{proof}
The first step applies to at most
$\zeta \|\cB\|$ atoms, since $\cB'$ $\zeta$-represents
$\cB$ with respect to each $f^{(1)}, \dots, f^{(R)}$. By
\cref{atomsize}, each atom occupies at most $\frac{1}{\|\cB\|} + \alpha(C)$ fraction of the entire
domain. So, the fraction of points whose values are set in the first
step is at most $\zeta \|\cB\|  (\frac{1}{\|\cB\|} +
\alpha(C))  < 2\zeta$.

In the second step, if
$\Pr[f(x)=f(z)~|~\cB'(x) = (c,s)]<\zeta$, then 
$\Pr[f(x)=f(z)~|~\cB(x) = c]<\Pr[f(x)=f(z)~|~\cB'(x) = (c,s)]+\zeta < 2\zeta$. Hence, the fraction of the
points  whose values are set in the second step is at most $2\zeta R = \eps/4$.

Thus, the distance of $F$ from $f$ is bounded by $2\zeta + \eps/4 < \eps/2$.
\end{proof}

\paragraph{Locating a violated constraint.}
We now want to use $F$ to ``find'' the affine constraint induced in
$f$. Setting $\mv{d}=(d_1,\dots,d_C)$ and
$\mv{k}=(k_1,\dots,k_C)$, we have by \cref{rem:pinduce} that the big picture function $F_{\cB}$ of $F$
will partially $(\mv{d},\mv{k})$-induce some constraint from $\cA$,
and hence by \cref{rem:induce}, it will partially $(\mv{d},
\mv{k})$-induce some $(A,\sigma) \in \cA$ of size $m \leq
\Psi_\cA(C,d)$ on $\ell$ variables. We will show that the original function $f$ violates many instances of this constraint.

 Denote the affine
forms in $A$ by $(L_1, \dots, L_m)$ and the vector $\sigma$ by
$(\sigma_1, \dots, \sigma_m)$. Since we can assume
$\ell\leq m$ (without loss of generality by making a change of
variables), we can now define
\begin{equation}\label{eqn:ell}
\ell_{\cA}(\eps)=\Psi_\cA(C_{\ref{thm:subatom2}}(\Delta, \eta,  \rho, \zeta, R),d).
\end{equation}

Let $b_1,\ldots,b_m \in \prod_{i=1}^C
\U_{k_i+1}$ correspond to the atoms of $\cB$ where $(A,\sigma)$ is
partially $(\mv{d}, \mv{k})$-induced by $F_{\cB}$. That is, $b_1,\ldots,b_m$ are consistent
with $A$, and $\sigma_i\in F_{\cB}(b_i)$ for every $i \in
[m]$. Also, let $b_1',\dots,b_{m}' \in \prod_{i = 1}^{C'} \U_{k_i +
  1}$ index the associated subatoms in $\cB'$, obtained by letting
$b_j' = (b_j,s)$ for every $j \in [m]$. 

\begin{lemma}\label{subatomconsistent}
The subatoms $b_1',\dots, b_m'$ are consistent with $A$.
\end{lemma}
\begin{proof}
Since $b_1, \dots, b_m$ are already consistent with $A$, we only need
to show that for every $i \in [C+1,C']$, the sequence $(b'_{1,i},
\dots, b'_{m,i}) = (s_{i-C}, s_{i-C}, \dots, s_{i-C})$ is $(d_i,
k_i)$-consistent. This holds because a constant function is of degree
$\leq d_i$.
\end{proof}

\paragraph{The main analysis.}
Let $\mv{x}=(x_1,\ldots,x_{\ell})$ where $x_1,\ldots,x_\ell$ are independent random variables taking values in $\F^n$ uniformly.   
Our goal is to prove a lower bound on
\begin{equation}\label{eqn:obj}
\Pr_{\mv{x}}\left[f(L_1(\mv{x}))=\sigma_1
  \wedge\cdots\wedge f(L_m(\mv{x}))=\sigma_m\right] 
=\E_{\mv{x}}\left[f^{(\sigma_1)}(L_1(\mv{x}))\cdots
  f^{(\sigma_m)}(L_m(\mv{x}))\right].
\end{equation}
The theorem obviously follows if the above expectation is larger than the respective $\delta_\cA(\eps)$. We rewrite the
expectation as
\begin{equation}\label{eqn:obj2}
\E_{\mv{x}}\left[(f_1^{(\sigma_1)}+
  f_2^{(\sigma_1)} + f_3^{(\sigma_1)})(L_1(\mv{x}))\cdots
  (f_1^{(\sigma_m)}+  f_2^{(\sigma_m)} + f_3^{(\sigma_m)})(L_m(\mv{x}))\right].
\end{equation}

We can expand the expression inside the expectation as a sum of $3^m$
terms. The expectation of any term involving $f_2^{(\sigma_j)}$
for any $j \in [m]$ is bounded in magnitude by 
$\|f_2^{(\sigma_j)}\|_{U^{d+1}} \leq \eta(|\cB'|)$, by
\cref{gowerscount} and the fact that the complexity of $A$ is
bounded by $d$. Hence, the expression (\ref{eqn:obj2}) is at least
\begin{equation*}
\E_{\mv{x}}\left[(f_1^{(\sigma_1)}+
  f_3^{(\sigma_1)})(L_1(\mv{x}))\cdots
  (f_1^{(\sigma_m)}+  f_3^{(\sigma_m)})(L_m(\mv{x}))\right]
- 3^{m} \eta(|\cB'|).
\end{equation*}
Now, because of the non-negativity of $f_1^{(\sigma_j)}+f_3^{(\sigma_j)}$
for every $j \in [m]$, this is at
least
\begin{equation}\label{eqn:obj3}
\E_{\mv{x}}\left[
\left(f_1^{(\sigma_1)}+  f_3^{(\sigma_1)}\right)(L_1(\mv{x}))\cdots
  \left(f_1^{(\sigma_m)}+  f_3^{(\sigma_m)}\right)(L_m(\mv{x}))
\prod_{j \in [m]}1_{[\cB'(L_j)=b_j']} 
\right]- 3^{m} \eta(|\cB'|),
\end{equation}
where $1_{[\cB'(L_j)=b_j']}$ is the indicator function of the event $\cB'(L_j(\mv{x}))=b_j'$.
In other words, now we are only counting  patterns that arise from the selected subatoms $b_1',
\dots, b_m'$. We next expand the product inside the expectation into $2^m$
terms. We will show that the contribution from each of the 
$2^m - 1$ terms involving  $f_3^{(\sigma_k)}$ for any $k \in [m]$ is small. Such a term is trivially bounded from above by
\begin{equation}\label{eqn:obj6}
\E_{\mv{x}}\left[ \left|f_3^{(\sigma_k)}(L_k(\mv{x})) \right|
    \prod_{j \in [m]}1_{[\cB'(L_j)=b_j']})
\right].
\end{equation}
Without loss of generality, we assume that $k=1$. This is convenient as by \cref{defaffine}~(i) we have $L_1(\mv{x})=x_1$. (For other values of $k$, we can do a change of variables, replacing $x_1$ with $L_k(\mv{x})$, so that we can assume $L_k(\mv{x})=x_1$.)
With the assumption $k=1$, the  square of (\ref{eqn:obj6}) is equal to the following. 
\begin{equation}\label{eq:afterCS}
\left(\E_{\mv{x}} \left[ \left|f_3^{(\sigma_k)}(x_1) \right|
   \prod_{j \in [m]}1_{[\cB'(L_j)=b_j']}
\right] \right)^2 \leq \E_{x_1}\left[|f_3^{(\sigma_k)}(x_1)|^2 
 1_{[\cB'(L_k)=b_k']} \right]
\E_{x_1}\left(\E_{x_2,\dots,x_\ell}\prod_{j \in [m]}1_{[\cB'(L_j)=b_j']}\right)^2.
\end{equation}
By \cref{thm:subatom2}~(vi) and \cref{atomsize}, we have  
\begin{equation}
\label{eq:firstTermafterCS}
\E_{x_1}\left[|f_3^{(\sigma_k)}(x_1)|^2 
 1_{[\cB'(L_k)=b_k']} \right] \le \Delta^2(C) \Pr_{x_1}[\cB'(x_1) = b_k'] \le \Delta^2(C)  \left(\frac{1}{\|\cB'\|}+\alpha(C')\right) \le \frac{2 \Delta^2(C)}{\|\cB'\|}.
\end{equation}
Let $\mv{y}=(y_2,\ldots,y_{\ell})$ where $y_2,\ldots,y_\ell$ are independent random variables taking values in $\F^n$ uniformly.  The second term in the right hand side of~(\ref{eq:afterCS}) is equal to
{\allowdisplaybreaks
\begin{align}
&\frac{1}{\|\cB'\|^{2m}}\E_{x_1}\left[\left(\E_{x_2,\dots,x_\ell}\prod_{i \in [C'] \atop j \in
    [m]}
  \frac{1}{p^{k_i+1}}\sum_{\lambda_{i,j}=0}^{p^{k_i+1}-1}\expo{\lambda_{i,j}\cdot(P_i(L_j'(\mv{x}))-b'_{i,j})}\right)^2\right] \nonumber\\
&=\frac{1}{\|\cB'\|^{2m}} \E_{x_1}\left[\left(\sum_{(\lambda_{i,j}) \in 
\atop \prod_{i, j} [0,p^{k_i + 1}-1]} \expo{-\sum_{i \in [C']\atop j \in
      [m]} \lambda_{i,j}b'_{i,j}} \E_{x_2,\dots,x_\ell}\expo{\sum_{i \in [C']\atop j \in
      [m]} \lambda_{i,j}P_i(L_j(\mv{x}))}\right)^2\right]\nonumber\\
&\leq \frac{1}{\|\cB'\|^{2m}}\sum_{(\lambda_{i,j}), (\tau_{i,j}) \in
\atop \prod_{i, j} [0,p^{k_i + 1}-1]} \left|\E_{\mv{x}, \mv{y}}\left[\expo{\sum_{i \in [C']\atop j \in
      [m]}
    \lambda_{i,j}P_i(L_j(\mv{x}))}\expo{-\sum_{i \in [C']\atop j \in
      [m]}
    \tau_{i,j}P_i(L_j(x_1,\mv{y}))}\right] \right|.\label{eqn:obj7}
\end{align} }

We can bound the above using \cref{dich}. Let $A'$ denote the set of $2m$
linear forms: $\set{L_j(x_1,x_2,\dots, x_\ell) \mid j \in [m]} \cup
\set{L_j(x_1,y_2,   \dots, y_\ell) \mid j \in [m]}$ in variables $x_1,\ldots,x_\ell,y_2,\ldots,y_\ell$. Let $\Lambda_i$
and $\Lambda_i'$ denote the $(d_i,k_i)$-dependency set of $A$ and
$A'$  respectively.
\begin{lemma}
For each $i$, $|\Lambda_i'| = |\Lambda_i|^2\cdot p^{k_i+1}$
\end{lemma}
\begin{proof}
Recall that  $L_1(\mv{x}) = L_1(x_1,\mv{y}) = x_1$. For any $\lambda, \tau \in \Lambda_i$ and any
$\alpha \in [0, p^{k_i+1}-1]$, note that $(\lambda_1 + \alpha \pmod {p^{k_i+1}},
\lambda_2, \dots, \lambda_m, \tau_1 -\alpha \pmod {p^{k_i+1}}, \tau_2,
\dots, \tau_m) \in \Lambda_i'$. Hence,
$|\Lambda_i'| \geq |\Lambda_i|^2 \cdot p^{k_i+1}$.
To show
$|\Lambda_i'| \leq |\Lambda_i|^2 \cdot p^{k_i+1}$, we give a map from
$\Lambda_i'$ to $\Lambda_i \times \Lambda_i$ that is
$p^{k_i+1}$-to-$1$. Suppose $\sum_{j=1}^m \lambda_j
Q(L_j(x_1,x_2,\dots,x_\ell))+\sum_{j=1}^m \tau_j
Q(L_j(x_1,y_2,\dots,y_\ell))\equiv 0$ for every polynomial $Q$ of
degree $d_i$ and depth $k_i$. Setting $x_2=\ldots=x_\ell=0$ shows that 
$$\sum_{j=1}^m \tau_j Q(L_j(x_1,y_2,\dots,y_\ell)) = -\left(\sum_{j=1}^m \lambda_j \right) Q(x_1),$$ and similarly setting $y_2=\ldots=y_\ell=0$ shows 
$$\sum_{j=1}^m \lambda_j Q(L_j(x_1,x_2,\dots,x_\ell))=-\left(\sum_{j=1}^m \tau_j \right) Q(x_1).$$ 
In particular $\sum_{j=1}^m \lambda_j= -\sum_{j=1}^m \tau_j$. Consequently,
$$(\lambda,\tau)  \mapsto \left(\left(-\sum_{j=2}^m \lambda_j \pmod {p^{k_i+1}}, \lambda_2, \dots, \lambda_m\right),
\left(-\sum_{j=2}^m \tau_i \pmod {p^{k_i+1}}, \tau_2, \dots, \tau_m\right)\right)$$ is a map from $\Lambda_i'$ to $\Lambda_i
\times \Lambda_i$. To see that it is $p^{k_i+1}$-to-$1$, note that $$(\lambda_1+\tau_1-\gamma \pmod {p^{k_i+1}},\lambda_2,\ldots,\lambda_m,\gamma,\tau_2,\ldots,\tau_m) \in \Lambda'_i$$ for every $\gamma \in [0,p^{k_i+1}-1]$, and these elements are all mapped to the same element in $\Lambda_i
\times \Lambda_i$. 
\end{proof}

Applying \cref{dich} (just as in the proof of \cref{affequid}), we get
that 
$$(\ref{eqn:obj7}) \le \frac{1}{\|\cB'\|^{2m}} \left(\prod_{i=1}^{C'} |\Lambda_i|^2
  p^{k_i+1} + \|\cB'\|^{2m}\alpha(C')\right)
\leq \frac{\prod_{i=1}^{C'}
    |\Lambda_i|^2}{\|\cB'\|^{2m}} + \alpha(C').$$
Combining this with \cref{eq:firstTermafterCS} and  \cref{eq:afterCS}, we obtain
\begin{align}\label{eqn:obj8}
(\ref{eqn:obj6}) \le  2\Delta(C)\sqrt{\frac{\prod_{i=1}^{C'}
    |\Lambda_i|^2}{\|\cB'\|^{2m}} + \alpha(C')}.
\end{align}

Finally, we turn to  the main term in the expansion of \cref{eqn:obj3}. We know from \cref{subatomconsistent} that the subatoms $b_1',\dots, b_m'$ are consistent with $A$. Thus 
\begin{align}
&\E_{\mv{x}}\left[
f_1^{(\sigma_1)}(L_1(\mv{x}))\cdots
 f_1^{(\sigma_m)}(L_m(\mv{x}))\cdot
\prod_{j \in [m]}1_{[\cB'(L_j)=b_j']}
\right]  \nonumber\\
&= \Pr[\cB'(L_1(\mv{x})) = b'_1\wedge\cdots\wedge\cB'(L_m(\mv{x})) = b'_m] \cdot \nonumber\\
&\qquad\qquad \E_{\mv{x}}\left[
f_1^{(\sigma_1)}(L_1(\mv{x}))\cdots
 f_1^{(\sigma_m)}(L_m(\mv{x}))|
\forall j\in [m], ~ \cB'(L_j(\mv{x})) = b'_j
\right]\nonumber\\
&\geq \left(\frac{\prod_{i=1}^{C'}|\Lambda_i|}{\|\cB'\|^m}
  -\alpha(C')\right)
\zeta^m.\label{eqn:obj9}
\end{align}
Let us justify the last line. The first term is due to
 the lower bound on the probability from Theorem
\ref{affequid}. The second term in (\ref{eqn:obj9}) follows since  each $f_1^{(\sigma_j)}$ is constant on the atoms of $\cB'$,
and because by construction, the big picture function $F_{\cB}$ of the
cleanup function $F$, on which $(A,\sigma)$ was partially induced,
supports a value inside an atom $b$ of $\cB$ only if the original
function $f$ acquires the value on at least an $\zeta$ fraction of
the subatom $(c,s)$.

Setting $\beta =
(\nfrac{\prod_{i=1}^{C'}|\Lambda_i|}{\|\cB'\|})^m$ and
combining the bounds from (\ref{eqn:obj3}), (\ref{eqn:obj8}) and
(\ref{eqn:obj9}), we conclude
\begin{align*}
(\ref{eqn:obj}) &\ge (\beta-\alpha(C'))\cdot \left(\frac{\eps}{8R}\right)^m -
2^{m+1}\Delta(C)\sqrt{\beta^2 + \alpha(C')} - 3^m\cdot
\eta(C')\\
&>
\frac{\beta}{2}\cdot\left(\frac{\eps}{8R}\right)^{\Psi_\cA(C,d)} -
2^{\Psi_\cA(C,d)+1}\beta \cdot \Delta(C) - 3^{\Psi_\cA(C,d)}\cdot
\eta(C')
\end{align*}
Since $\|\cB'\|\leq p^{dC'}$, $\frac{1}{\|\cB'\|^{\Psi_\cA(C,d)}} \leq
\beta \leq 1$, $\Delta(C) =
\frac{1}{16}(\frac{\eps}{8R})^{\Psi_{\cA}(d,C)}$, $\eta(C') <
\frac{1}{8 \|\cB'\|^{\Psi_\cA(C,d)}}
\left(\frac{\eps}{24R}\right)^{\Psi_\cA(C,d)}$, and
 both $C$ and $C'$ are upper-bounded by
 $C_{\ref{thm:subatom2}}(\Delta, \eta,  \rho, \zeta, R)$, we can now
 define
\begin{equation}\label{eqn:delta}
\delta_{\cA}(\eps)=\frac14p^{-d\Psi_A(C_{\ref{thm:subatom2}}(\Delta,
  \eta, \rho, \zeta, R))C_{\ref{thm:subatom2}}(\Delta, \eta, \rho,
  \zeta,
  R)}\cdot\left(\frac{\eps}{8R}\right)^{\Psi_\cA(C_{\ref{thm:subatom2}}(\Delta,
  \eta, \rho, \zeta, R),d)}
\end{equation}
to conclude the proof.

\end{proofof}

\bibliographystyle{alpha}
\bibliography{testing}

\end{document}